\documentclass[a4paper,11pt]{amsart}
\usepackage{amssymb,amsmath,amsthm}
\usepackage{graphicx,color}
\usepackage{stackrel}
\usepackage{hyperref}
\newtheorem{Lemma}{Lemma}
\newtheorem{Theorem}{Theorem}
\newtheorem{Corollary}{Corollary}
\newtheorem{Proposition}{Proposition}

\newtheorem{Definition}{Definition}
\newtheorem{Example}{Example}

\theoremstyle{definition}

\numberwithin{equation}{section}
\usepackage[normalem]{ulem}

\definecolor{DarkGreen}{rgb}{0,0.5,0.1}

\newcommand\soutD{\bgroup\markoverwith
{\textcolor{DarkGreen}{\rule[.5ex]{2pt}{1pt}}}\ULon}
\newcommand\soutP{\bgroup\markoverwith
{\textcolor{blue}{\rule[.5ex]{2pt}{1pt}}}\ULon}
\newcommand{\Hm}[1]{\leavevmode{\marginpar{\tiny%
$\hbox to 0mm{\hspace*{-0.5mm}$\leftarrow$\hss}%
\vcenter{\vrule depth 0.1mm height 0.1mm width \the\marginparwidth}%
\hbox to
0mm{\hss$\rightarrow$\hspace*{-0.5mm}}$\\\relax\raggedright #1}}}

\begin{document}
%
\title[Uncertainty relations for noncommutative algebra]{Uncertainty relations for a non-canonical phase-space noncommutative algebra}

\author[Dias and Prata]{Nuno C. Dias and Jo\~{a}o N. Prata}

\address{Grupo de F\'{\i}sica Matem\'{a}tica, Faculdade de Ci\^{e}ncias da Universidade de Lisboa, Campo Grande, Edif\'{\i}cio C6
1749-016 Lisboa, Portugal, and Escola Superior Náutica Infante D. Henrique, Av. Engenheiro Bonneville Franco, 2770-058 Pa\c{c}o de Arcos, Portugal} \email{ncdias@meo.pt}
\email{jnprata@FC.UL.PT}

\date{\today}

\begin{abstract}
We consider a non-canonical phase-space deformation of the Heisenberg-Weyl algebra that was recently introduced in the context of quantum cosmology. We prove the existence of minimal uncertainties for all pairs of non-commuting variables. We also show that the states which minimize each uncertainty inequality are ground states of certain positive operators. The algebra is shown to be stable and to violate the usual Heisenberg-Pauli-Weyl inequality for position and momentum. The techniques used are potentially interesting in the context of time-frequency analysis.
\end{abstract}

\maketitle

\section{Introduction}

Noncommutative geometry (NCG) is considered to be a fundamental feature of space-time at the Planck scale. Indeed, configuration space noncommutativity arises when one considers the low energy effective theory of a D-brane in the background of a Neveu-Schwartz B field \cite{Seiberg}. This fact has triggered the investigation of what qualitative and quantitative effects may appear when one adds extra (phase-space or configuration space) noncommutativity to the traditional position-momentum ones. Various aspects of such theories have been investigated in the context of quantum gravity and string theory \cite{Bertolami1,Connes1,Madore,Seiberg}, quantum field theory \cite{Carroll,Connes2,Douglas,Szabo}, non-relativistic quantum mechanics \cite{Bertolami2,Bertolami3,Bastos1,Bastos9,Bolonek,Dey1,Dey2,Dey3,Dey4,Delduc,Kochan,Gamboa,Kosinski,Matoiu,Nair,Lein}, quantum Hall effect \cite{Bellissard,Duval,Horvathy}, condensed matter \cite{Isidro}, and quantum cosmology \cite{Bastos5,Bastos6,Bastos7,Obregon,Farhoudi,Nicolini}.

The additional noncommutativity is regarded as a deformation of the Poincar\'e or of the Heisenberg-Weyl (HW) algebra. Here we shall consider the latter case. The deformation may be of various natures. It may be canonical (in the sense that the commutators are equal to c-numbers) or non-canonical. It may affect only the configuration sector, only the momentum sector, or both. The former cases will henceforth be denoted as one-sector deformations, whereas the latter case is dubbed a phase-space deformation. The type of deformation may be dictated by compelling physical arguments (such as configuration space noncommutativity in the case of string theory) or by mathematical arguments related to the consistency of the theory (e.g. stability of the algebra \cite{Faddeev,Vilela}). One-sector deformations break the symmetry between position and momentum found in ordinary quantum and classical mechanics. On the other hand, phase-space noncommutativity has some unexpected physical implications in the context of quantum mechanics, quantum cosmology and black hole (BH) physics.

Here we wish to investigate further a non-canonical phase-space deformation of the Heisenberg-Weyl algebra introduced in \cite{Bastos7,Bastos8}. Our interest in this algebra is twofold: there are physical and mathematical motivations. The physical motivation comes from the fact that phase-space noncommutativity seems to be a necessary ingredient for the thermodynamical stability of BHs \cite{Bastos6,Nicolini}, and may also contribute to the regularization of singularities \cite{Bastos7,Bastos8,Nicolini}. As a rule of thumb, two noncommuting variables satisfy uncertainty principles which preclude a sharp simultaneous localization of both variables. It is this delocalization which regularizes the BH singularity. A canonical phase-space noncommutative algebra is a step in the direction of some smoothing but not complete regularization of the singularity \cite{Bastos5,Bastos6}. A full-fledged regularization was accomplished with our non-canonical phase-space noncommutative algebra \cite{Bastos7,Bastos8}.

Here is a brief sketch of how this was achieved. In \cite{Bastos5,Bastos6,Obregon}, the Heisenberg-Weyl algebra
\begin{equation}
\left[ \widehat{q}_1,\widehat{p}_1 \right]=\left[ \widehat{q}_2,\widehat{p}_2 \right]=i,
\label{eqCorrections2}
\end{equation}
is replaced by the following canonical deformation:
\begin{equation}
\left[ \widehat{q}_1,\widehat{p}_1 \right]=\left[ \widehat{q}_2,\widehat{p}_2 \right]=i, \qquad \left[ \widehat{q}_1,\widehat{q}_2 \right]= i \theta, \qquad \left[ \widehat{p}_1,\widehat{p}_2 \right]=  i \eta,
\label{eqCorrections2}
\end{equation}
where all the remaining commutators vanish and $\theta, \eta $ are some constants which are assumed to be small ($\theta, \eta <<1$). The Wheeler-De Witt equation (WDW) for the Kantowski-Sachs black hole \cite{Kantowski} is given by (after a particular choice of operator order):
\begin{equation}
\left(\widehat{p}_1^2-\widehat{p}_2^2 -48 e^{-2 \sqrt{3}\widehat{q}_2}\right) \psi=0.
\label{eqCorrections3}
\end{equation}
where the configuration variables $q_1,q_2$ are the scale factors of the KS metric.

With the usual differential representation for the Heisenberg-Weyl algebra
\begin{equation}
\widehat{q}_1=x_1, \qquad \widehat{q}_2=x_2 , \qquad \widehat{p}_1= -i \frac{\partial}{\partial x_1} , \qquad \widehat{p}_2= -i \frac{\partial}{\partial x_2},
\label{eqCorrections4}
\end{equation}
the WDW equation (\ref{eqCorrections3}) reads:
\begin{equation}
\left(\frac{\partial^2}{\partial x_2^2} - \frac{\partial^2}{\partial x_1^2} -48 e^{-2 \sqrt{3} x_2} \right) \psi(x_1,x_2)=0,
\label{eqCorrections5}
\end{equation}
with solutions of the form
\begin{equation}
\psi(x_1,x_2) =\psi_{\nu}^{\pm} (x_1,x_2)=e^{\pm i \nu \sqrt{3}x_1}K_{i \nu}\left(4e^{- \sqrt{3}x_2}\right),
\label{eqCorrections6}
\end{equation}
where $K_{i \nu}$ are modified Bessel functions. These solutions are highly oscillatory and not square-integrable. This poses severe interpretational problems. This is a familiar feature in this type of mini superspace models. One faces the problem of determining a "time" variable and a measure, such that on constant "time" hypersurfaces, the wave-function is normalizable and the square of its modulus is a bona fide probability density.

On the other hand, with the deformation (\ref{eqCorrections2}) and its differential representation
\begin{equation}
\left\{
\begin{array}{l}
\widehat{q}_1 = \lambda x_1 - \frac{i \theta}{2 \lambda} \frac{\partial}{\partial x_2}\\
\\
\widehat{q}_2 = \lambda x_2 + \frac{i \theta}{2 \lambda} \frac{\partial}{\partial x_1}\\
\\
\widehat{p}_1 = -i \mu \frac{\partial}{\partial x_1} - \frac{\eta}{2 \mu} x_2 \\
\\
\widehat{p}_2 = -i \mu \frac{\partial}{\partial x_2} + \frac{\eta}{2 \mu} x_1
\end{array}
\right.
\label{eqCorrections7}
\end{equation}
where $\mu, \lambda$ are dimensionless constants such that $2 \lambda \mu= 1+ \sqrt{1- \theta \eta}$, the WDW equation becomes:
\begin{equation}
\begin{array}{c}
\left\{ \left(i \mu \frac{\partial}{\partial x_1} + \frac{\eta x_2}{2 \mu} \right)^2 -\left(i \mu \frac{\partial}{\partial x_2} - \frac{\eta x_1}{2 \mu} \right)^2\right.\\
 \\
 \left. -48 \exp \left[ -2 \sqrt{3} \left(\lambda x_2 + \frac{i \theta}{2 \lambda}\frac{\partial}{\partial x_1}   \right) \right]\right\} \psi(x_1,x_2)=0
\end{array}
\label{eqCorrections8}
\end{equation}
The solutions of this equation are of the form
 \begin{equation}
 \psi_a(x_1,x_2)= \mathcal{R}_a (x_2)\exp \left[ \frac{i x_1}{\mu} \left(a-\frac{\eta}{2 \mu} x_2 \right) \right],
\label{eqCorrections8.A}
\end{equation}
where $a$ is an arbitrary real constant and $\phi_a(x)=\mathcal{R}_a \left(\mu x + \frac{\theta a}{2 \lambda} \right)$ satisfies the time-independent Schr\"odinger equation $- \phi^{\prime \prime} (x) + V(x) \phi (x)=0$, with potential
\begin{equation}
V(x)=48 e^{-2 \sqrt{3} x} -(\eta x-c)^2, \hspace{1 cm} c \in \mathbb{R}.
\label{eqCorrections8.B}
\end{equation}
The solutions are still not square integrable. However, one can observe a distinct dampening of the amplitude of the oscillations. It is also worth noting that the noncommutativity in the momentum sector ($\eta$) leads to the existence of a stable minimum of the potential and consequently to the thermodynamic stability of the black hole \cite{Bastos6}.

In \cite{Bastos7,Bastos8} we suggested a non-canonical noncommutative deformation of the Heisenberg-Weyl algebra (see section 3) which leads to another WDW equation. After a separation of variables akin to (\ref{eqCorrections8.A}), one obtains the time-independent Schr\"odinger equation, this time with potential
\begin{equation}
\begin{array}{c}
V(x)=-(\eta x-a)^2-F^2\mu^4x^4-2 F \mu^2 (\eta x-a)x^2+\\
\\
+48 \exp \left(-2 \sqrt{3} x-2 \sqrt{3} \mu^2Ex^2+ \frac{\sqrt{3} \theta a}{\mu \lambda} \right),
\end{array}
\label{eqCorrections9}
\end{equation}
where $F$ and $E$ are certain constants related to the algebra.

As previously, this potential also exhibits a stable minimum \cite{Bastos8}. But, more importantly, the asymptotically dominant term $V(x) \sim - F^2\mu^4 x^4$, for $z \to \infty$, leads to square integrable solutions of the WDW equation. This then permits the evaluation, as in ordinary quantum mechanics, of probabilities according to Born's rule. If we compute the probability of finding the scale factors near the singularity of the Kantowski-Sachs black hole, we conclude that the probability vanishes (see \cite{Bastos7}). Thus, in this case, the singularity is not "erased" by the existence of some minimum length as suggested by various authors \cite{Dey1,Dey2,Dey3,Dey4,Fityo,Kempf,Kober}. Rather, the singularity is still there, but the probability of reaching it is zero.

\vspace{0.5 cm}
On the other hand, the mathematical motivations are the following. Our algebra \cite{Bastos7,Bastos8} seems to be a minimal departure from the canonical phase-space noncommutative algebra in the sense that:

\vspace{0.3 cm}
\begin{itemize}
\item It is a parsimonious deformation as it only introduces one additional deformation parameter accounting for the non-canonical nature of the algebra.

\item It is isomorphic with the usual Heisenberg-Weyl algebra.
\end{itemize}

\vspace{0.3 cm}
The latter property ensures the stability of our algebra (see below). However, in spite of its simple nature, it displays several interesting features:

\vspace{0.3 cm}
\begin{enumerate}
 \item All pairs of noncommuting variables satisfy uncertainty relations. We provide a means of obtaining the sharp constants and the corresponding minimizers, which are solutions of certain partial differential equations.

\item Contrary to what happens in ordinary quantum mechanics, it seems that there are no quantum states saturating more than one of the uncertainty relations simultaneously (this seems to be a common feature of noncommutative extensions of the Heisenberg-Weyl algebra \cite{Bolonek,Kosinski})

\item The usual position-momentum uncertainty relation may be violated.

\item There are no minimal length and momentum.
\end{enumerate}

\vspace{0.3 cm}
Items (1) and (2) above will help clarify the following important issue. The noncanonical algebra was quite successful at regularizing the singularity of a Kantowski-Sachs black hole \cite{Bastos7,Bastos8}. It would therefore be important to determine whether there are states of minimal uncertainty. We will prove that there are states which minimize individual uncertainty relations, but we will argue that it does not seem possible to find states which saturate {\it all} uncertainty relations simultaneously.

The techniques used to prove that there are minimizers for the various uncertainty relations and to obtain certain equations satisfied by these minimizers come from variational calculus \cite{Evans,Jahn,Jost} and compact embedding theorems for a class of functional spaces called {\it modulation spaces} \cite{Feichtinger,Grochenig}. The embedding theorems are due to Boggiatto and Toft \cite{Boggiatto}, Pfeuffer and Toft \cite{Pfeuffer} and they can also be viewed as isomorphisms of functional spaces via certain Toeplitz localization operators \cite{GrochenigToft1,GrochenigToft2}.

The uncertainty principles that we obtain can be related to the uncertainty principles of Cowling and Price \cite{Cowling}, in the sense that one considers weights other than the ones leading to the covariance of position and momentum as in Heisenberg's uncertainty principle. However, we go one step further in the sense that the weights mix position $(x)$ and momentum $(\xi= -i \partial_x)$:
\begin{equation}
\| u(x,-i \partial_x) f\|_{L^2}^2 + \|v(x,- i\partial_x)f\|_{L^2 }^2 \geq C \|f \|_{L^2}^2
\label{eqIntroduction1}
\end{equation}
for some constant $C>0$. Moreover, we show that there are minimizers. This is in contrast with \cite{Cowling}, where only the existence of an {\it infimum} is proved.

Moreover these uncertainty principles can also be related to continuous embedding theorems of functional spaces in the spirit of \cite{Galperin1,Galperin2,GrochenigStudia}.

In this work we shall consider units $\hbar=1$.

\section*{Notation}

The variable $x=(x_1, \cdots, x_d)$ denotes a generic point in $\mathbb{R}^d$ representing a position variable, whereas $\xi =(\xi_1, \cdots, \xi_d) \in \mathbb{R}^d$ denotes the momentum. The usual scalar product in $\mathbb{R}^d$ is denoted by $u \cdot v= \sum_{i=1}^d u_i v_i$ or $u \cdot v$ for $u,v \in \mathbb{R}^d$ and the corresponding norm is $|u|= \left(\sum_{i=1}^d u_i^2 \right)^{\frac{1}{2}}$.

$\mathcal{S} (\mathbb{R}^d)$ is the Schwartz space of test functions and its dual $\mathcal{S}^{\prime} (\mathbb{R}^d)$ is the space of tempered distributions. $< \cdot, \cdot>$ is the distributional bracket $\mathcal{S}^{\prime} (\mathbb{R}^d) \times \mathcal{S} (\mathbb{R}^d) \to \mathbb{C}$. Given a Hilbert space $\mathcal{H}$, the inner product is denoted by $\langle \cdot, \cdot \rangle_{\mathcal{H}}$, which we assume to be linear in the first argument and anti-linear in the second. The corresponding norm is $\| f \|_{\mathcal{H}}^2= \langle f,f \rangle_{\mathcal{H}}$.

The Fourier transform of a function $f(x) \in L^2(\mathbb{R}^d)$ is denoted $\widetilde{f}(\xi)$ and is given (as a limiting process of functions in $L^1(\mathbb{R}^d) \cap L^2(\mathbb{R}^d)$) by:
\begin{equation}
\widetilde{f} (\xi) := (2 \pi )^{- d/2} \int_{\mathbb{R}^d} f(x) e^{- i x \cdot \xi} dx.
\label{eq0.1}
\end{equation}
Notice that we are using the physicists convention rather than the usual definition in harmonic analysis:
\begin{equation}
\mathcal{F} f( \omega)  := \int_{\mathbb{R}^d} f(t) e^{-2 \pi i t \cdot \omega} dt
\label{eq0.2}
\end{equation}
where $t$ is "time" and $\omega$ is "frequency".

If there is a positive constant $C>0$ such that $A \leq C B$, we write $A \lesssim B  $. If $A \lesssim B$ and $B \lesssim A$, then we shall simply write $A \asymp B$.

A generic operator acting on a Hilbert space $\mathcal{H}$ is denoted by $\widehat{A}$, its adjoint is $\widehat{A}^{\ast}$, its domain, range and kernel are $\text{Dom} (\widehat{A})$, $\text{Ran} (\widehat{A})$ and $\text{Ker} (\widehat{A})$, respectively. Its operator norm is $\|\widehat{A}\|_{op} := \text{sup}_{\|f\|_{\mathcal{H}} \leq 1} \|\widehat{A} f \|_{\mathcal{H}}$.

A sequence $(f_n)_{n \in \mathbb{N}}$ in the Hilbert space $\mathcal{H}$ converges strongly to $f \in \mathcal{H}$, if $\|f_n -f \|_{\mathcal{H}} \to 0$ as $n \to \infty$. In this case, we write $f_n \to f$. Likewise, it converges weaky, if $\langle f_n -f, g \rangle_{\mathcal{H}} \to 0$ as $n \to \infty$, for all $g \in \mathcal{H}$, and we write $f_n \rightharpoonup f$.

We denote the compact embedding of a functional space $\mathcal{B}_1$ into another functional space $\mathcal{B}_2$ by $\mathcal{B}_1 \subset \subset \mathcal{B}_2$.

\section{Remarks on the uncertainty principle}

The Heisenberg uncertainty principle is one of the cornerstones of quantum mechanics, harmonic analysis and time-frequency analysis. Loosely speaking, it states that a simultaneous measurement of the position and momentum of a particle with infinite precision is precluded. This is in sharp contrast with the laws of classical mechanics, where hindrances to the precision of simultaneous measurements of any pair of observables can only be attributed to the quality of the measuring apparatuses. For a survey of mathematical aspects of the uncertainty principle see \cite{Folland2}. Good discussions on the physical interpretation and implications of the uncertainty principle can be found in \cite{Busch1,Busch2}.

The original paper of Heisenberg \cite{Heisenberg} begins with a famous discussion of the resolution of microscopes, in which the accuracy (resolution) of an approximate position measurement is related to the disturbance of the particle's momentum. It is quite remarkable that Heisenberg never gave a precise definition of what he meant by {\it resolution} and {\it disturbance}. In most textbooks on quantum mechanics, one is introduced to the version of Kennard \cite{Kennard}, Robertson \cite{Robertson} and Weyl \cite{Weyl}, where resolution and disturbance are understood as the mean standard deviations of the position (resp. momentum) with respect to the probability measure $|f(x)|^2 dx$ (resp. $|\widetilde{f} (\xi)|^2 d \xi$) for a given wave function $f \in L^2 (\mathbb{R})$. Denoting these quantities by $\Delta_x \left(f,<x>_f \right)$ and $\Delta_{\xi} \left(f,<\xi>_f \right)$, respectively (see the precise definitions below), they were able to prove the following inequality:
\begin{equation}
\Delta_x \left(f,<x>_f \right) ~ \Delta_{\xi} \left(f,<\xi>_f \right) \geq \frac{\|f\|_{L^2 (\mathbb{R})}^2}{2}.
\label{eqCorrectionsC1}
\end{equation}
Since the mean standard deviation is interpreted as a measure of the dispersion of a probability measure relative to its mean value, the previous inequality states that $f$ and $\widetilde{f}$ cannot be both sharply localized.

The fact that one used the mean standard deviation (or equivalently the variance) as the measure of dispersion is somewhat arbitrary. Other uncertainty principles use other quantities, like for instance the entropy, as a measure of dispersion. Shannon \cite{Shannon} proved that, for a given probability measure $\mu$ with covariance matrix $Cov(\mu)$ and entropy
\begin{equation}
E(\mu)=-\int_{\mathbb{R}} \mu(x) \log \left(\mu(x) \right) dx,
\label{eqCorrectionsC2}
\end{equation}
the following inequality holds:
\begin{equation}
E(\mu) \leq \frac{1}{2} \log \left[ 2 \pi e \det \left(Cov (\mu) \right) \right].
\label{eqCorrectionsC3}
\end{equation}
Beckner \cite{Beckner}, Bialynicki-Birula and Mycielski \cite{Birula}, and Hirschman \cite{Hirschman} proved the following {\it entropic} uncertainty principle:
\begin{equation}
\log(\pi e)\leq E \left(|f|^2 \right)+ E \left (|\widetilde{f}|^2 \right).
\label{eqCorrectionsC4}
\end{equation}

There are also several different ways by which we can combine the dispersions to obtain a measure of uncertainty. The variance is related to quantities such as $\|xf\|_{L^2 (\mathbb{R})}$ and $\|\xi \widetilde{f}\|_{L^2 (\mathbb{R})}$. If these measure the dispersion of $f$ and $\widetilde{f}$ relative to the origin, then a measure of uncertainty could be the product $\|xf\|_{L^2 (\mathbb{R})} ~ \|\xi \widetilde{f}\|_{L^2 (\mathbb{R})}$, as in (\ref{eqCorrectionsC1}). But we could also express it as $\|xf\|_{L^2 (\mathbb{R})}^2 + \|\xi \widetilde{f}\|_{L^2 (\mathbb{R})}^2$. Indeed the sum could be more useful than the product in certain cases. Suppose a certain quantity represented by an observable $\widehat{A}$ has a measure of dispersion (variance, entropy, or other) given by $\Delta_A=n$, while another observable $\widehat{B}$, which does not commute with $\widehat{A}$, has a measure of dispersion $\Delta_B=\frac{1}{n^2}$. Then the product of the dispersions is given by $\Delta_A ~ \Delta_B= \frac{1}{n}$, while $\Delta_A^2 + \Delta_B^2= n^2 + \frac{1}{n^4}$. If we could control the state in such a way that $n \to \infty$, then  $\Delta_A ~ \Delta_B \to 0$, while $\Delta_A^2 + \Delta_B^2 \to \infty$. The lesson from this example is that, if the observables $\widehat{A}$ and  $\widehat{B}$ are such that it is possible to find states for which $\Delta_A \to \infty$ and $\Delta_B \to 0 $ at a different pace, then the product $\Delta_A ~\Delta_B$ may not be a good measure of the uncertainty, as it can be made arbitrarily small. We will give briefly concrete examples for this. To circumvent these difficulties Cowling and Price \cite{Cowling,Folland2} have considered uncertainty principles of the form:
\begin{equation}
\|~ |x|^a f \|_{L^p (\mathbb{R})} + \|~ |\xi|^b \widetilde{f} \|_{L^q (\mathbb{R})} \geq K \|f \|_{L^2 (\mathbb{R})},
\label{eqCorrectionsC5}
\end{equation}
which hold for all $p,q \in  \left[1, \infty \right]$, all tempered functions $f$ such that $\widetilde{f} $ is also a function, and all $a,b >0$, such that:
\begin{equation}
a > \frac{1}{2}-\frac{1}{p} \hspace{1 cm} \text{ and } \hspace{1 cm}  b > \frac{1}{2}-\frac{1}{q}.
\label{eqCorrectionsC6}
\end{equation}

So basically, there are various measures of dispersion, and several different ways of combining them to obtain a mesure of uncertainty. Some measures can be more suitable than others to obtain bounds for the variance of particular observables.  

We will now try to specify a bit more the previous ideas and give some examples which illustrate that for arbitrary noncommuting observables $\widehat A$ and $\widehat B$, not all measures of uncertainty are equivalent and, in particular, 
the traditional uncertainty inequality (\ref{eqCorrectionsC1}) may sometimes fail to reveal that there is an uncertainty in the first place.

Generally speaking, if $\widehat{A}$ and $\widehat{B}$ are two non-commuting, essentially self-adjoint operators acting on some Hilbert space $\mathcal{H}$ with inner product $\langle \cdot, \cdot \rangle_{\mathcal{H}}$ and norm $\| \cdot \|_{\mathcal{H}}$, then
\begin{equation}
\| (\widehat{A} -a \widehat{I}) f \|_{\mathcal{H}} ~ \|(\widehat{B} -b \widehat{I}) f \|_{\mathcal{H}} \geq \frac{1}{2}\left|\langle \left[\widehat{A}, \widehat{B} \right] f, f \rangle_{\mathcal{H}} \right|
\label{eq1.1}
\end{equation}
for any $a,b \in \mathbb{R}$ and $f \in \text{Dom}(\widehat{A} \widehat{B}) \cap \text{Dom}(\widehat{B} \widehat{A})$. In the previous inequality $\widehat{I}$ denotes the identity operator in $\mathcal{H}$ and $\left[\widehat{A}, \widehat{B} \right] =\widehat{A} \widehat{B} - \widehat{B} \widehat{A}$ is the commutator. Uncertainty relations of this form can also be considered for non self-adjoint operators, like the ones appearing in PT-symmetric quantum mechanics \cite{Bender,Dey1}. In this work, however, we will only deal with essentially self-adjoint operators. Moreover, we shall always assume the states to be normalized $\|f\|_{\mathcal{H}}=1$.

The equality in (\ref{eq1.1}) holds for a given $f \in \mathcal{H}$ if and only if there exists a constant $c \in \mathbb{R}$ such that
\begin{equation}
(\widehat{A} -a \widehat{I}) f = i c (\widehat{B} -b \widehat{I}) f
\label{eq1.2}
\end{equation}
The Heisenberg uncertainty principle emerges if one considers the position and momentum of a particle. In that case $\widehat{A} =\widehat{x}= \mbox{multiplication by }x$ and $\widehat{B}= \widehat{\xi} = - i  \frac{d}{dx} $ acting on $\mathcal{H} = L^2 (\mathbb{R})$.
Substituting in (\ref{eq1.1}), we obtain:
\begin{equation}
\Delta_x (f,a)\Delta_{\xi} (f,b)\ge \frac{1}{2}
\label{eq1.3}
\end{equation}
where $\Delta_x(f,a)$ and $\Delta_{\xi} (f,b)$ are the position and momentum dispersions, respectively:
\begin{equation}
\begin{array}{l}
\Delta_x (f,a)= \|(\widehat{x} -a \widehat{I}) f \|_{L^2 (\mathbb{R})} = \left(\int_{\mathbb{R}} (x-a)^2 | f (x)|^2 dx \right)^{\frac{1}{2}}\\
\\
\Delta_{\xi} (f,b)= \|(\widehat{\xi} -b \widehat{I}) f \|_{L^2 (\mathbb{R})} = \left(\int_{\mathbb{R}} (\xi-b)^2 | \widetilde{f} (\xi)|^2 d \xi . \right)^{\frac{1}{2}}
\end{array}
\label{eq1.4}
\end{equation}
As usual the dispersion becomes minimal if we set $a=<x>_f,~b=<\xi>_f$, which are the expectation values of the position and momentum in the state $f$:
\begin{equation}
\begin{array}{l}
<x>_f= \langle \widehat{x}f, f \rangle_{L^2 (\mathbb{R})}= \int_{\mathbb{R}} x | f (x)|^2 dx \\
\\
<\xi>_f= \langle \widehat{\xi}f, f \rangle_{L^2 (\mathbb{R})}= \int_{\mathbb{R}} \xi | \widetilde{f} (\xi)|^2 d \xi .
\end{array}
\label{eq1.4.1}
\end{equation}
In this case $\Delta_x (f,<x>_f)$ and $\Delta_{\xi} (f,<\xi>_f)$ are called the mean standard deviations of position and momentum.

From (\ref{eq1.2}) equality holds in (\ref{eq1.3}) if and only if $f$ is a generalized Gaussian state.

There are various instances where certain observables, other than position or momentum, may be more relevant. For example, as we shall see below, noncommutative theories may lead to more intricate composite operators of position and/or momentum. Alternatively, one may be interested in energy rather than, say, momentum. Such cases will then require more general uncertainty principles. Inequality (\ref{eq1.1}) would then be a good starting point. However, unlike the case of the Heisenberg uncertainty principle, inequality (\ref{eq1.1}) may not lead in general to a useful uncertainty principle for given noncommuting operators $\widehat{A}$ and $\widehat{B}$. Indeed:

\vspace{0.3 cm}

\begin{enumerate}
\item Even if two operators $\widehat{A}, \widehat{B}$ are noncommuting, the product of their dispersions $\Delta_A (f,a) \Delta_B (f,b)$ need not be bounded from below by a positive constant for all normalized $f$.

\vspace{0.2 cm}
\item The commutator $\left[\widehat{A}, \widehat{B} \right]$ does not necessarily provide a lower positive bound on the product of the dispersions as stated in (\ref{eq1.1}).
\end{enumerate}

These are well known facts \cite{Folland2}, but to make our presentation self-contained and motivate an alternative formulation of the uncertainty principle, we will give simple examples that support these claims.

Let us start by showing that the product of the dispersions of the observables $(\widehat{x})^n$ and $(\widehat{\xi})^m$ with $n,m \in \mathbb{N}$ may be as close to zero as we wish, as long as $n \ne m$, even though they are noncommuting. Moreover, we shall give an example of two noncommuting observables and a non-zero state, such that the right-hand side of (\ref{eq1.1}) vanishes exactly. To keep our discussion simple, we will consider only one dimensional systems $(d=1)$ in this section.

We shall require the following unitary operator. It is called the dilation operator and plays an important role in signal processing \cite{Daubechies,Grochenig}:
\begin{equation}
D_s f (x) = \frac{1}{\sqrt{|s|}} f \left( \frac{x}{s} \right)
\label{eq2.1}
\end{equation}
for $s \in \mathbb{R} \backslash \left\{0 \right\}$. The Fourier transform acts as:
\begin{equation}
\widetilde{D_s f} (\xi) = D_{\frac{1}{s}} \widetilde{f} (\xi)
\label{eq2.2}
\end{equation}
Let us now consider the observables $\widehat{A}= (\widehat{x})^n$ and $\widehat{B}= (\widehat{\xi})^m$ for $n,m \in \mathbb{N}$ and some normalized state $f_1 \in \mathcal{S} (\mathbb{R})$. We set $a=b=0$ for the moment. We thus have:
\begin{equation}
\Delta_A (f_1,0) = \left(\int_{\mathbb{R}} x^{2n} | f_1 (x)|^2 dx \right)^{\frac{1}{2}}, \hspace{1 cm} \Delta_B (f_1,0) = \left(\int_{\mathbb{R}} \xi^{2m} | \widetilde{f_1} (\xi)|^2 d \xi \right)^{\frac{1}{2}}
\label{eq2.3}
\end{equation}
and an uncertainty
\begin{equation}
\Delta_A (f_1,0) \Delta_B (f_1,0) = C_1
\label{eq2.4}
\end{equation}
for some $C_1 >0$.

Next consider the state $f_s = D_s f_1$ for some $s \ne 0$. A simple calculation reveals that:
\begin{equation}
\begin{array}{c}
\Delta_A (f_s,0)= \left(\int_{\mathbb{R}} x^{2n} | f_s (x)|^2 dx\right)^{\frac{1}{2}} = |s|^{-1/2} \left(\int_{\mathbb{R}} x^{2n} | f_1 \left( s^{-1} x\right)|^2 dx \right)^{\frac{1}{2}}=\\
\\
=|s|^{-1/2} \left(\int_{\mathbb{R}} (sy)^{2n} | f_1 \left(y\right)|^2 |s|dy \right)^{\frac{1}{2}}= |s|^n \Delta_A (f_1,0)
\end{array}
\label{eq2.5}
\end{equation}
A similar calculation leads to
\begin{equation}
\Delta_B (f_s,0) = |s|^{-m} \Delta_B (f_1,0)
\label{eq2.6}
\end{equation}
Altogether, we obtain:
\begin{equation}
\Delta_A (f_s,0) \Delta_B (f_s,0) = |s|^{n-m} C_1
\label{eq2.7}
\end{equation}
If $n=m$ then the uncertainty is left unchanged. However, if, say, $m >n$, then as $s \to + \infty$, the uncertainty can be made arbitrarily small. If this holds for $a=b=0$, then it must also hold for the mean-standard deviations:
\begin{equation}
|s|^{n-m} C_1=\Delta_A (f_s,0) \Delta_B (f_s,0) \geq \Delta_A (f_s,<A>_{f_s}) \Delta_B (f_s,<B>_{f_s})  \to 0,
\label{eq2.7.A}
\end{equation}
as $s \to + \infty$.

This proves that noncommutativity may not necessarily pose an obstacle to arbitrarily small products of dispersions.

In particular, for instance the kinetic energy $(\widehat{\xi})^2$ and the position $\widehat{x}$ or any potential energy of the form $(\widehat{x})^n$ with $n \in \mathbb{N} \backslash \left\{2 \right\}$ can have arbitrarily small product of dispersions.

In contrast with this, the observables $(\widehat{x})^n$ and $(\widehat{\xi})^n$ do have a positive minimum product uncertainty \cite{Cowling,Hirschman}. The case $n=m=1$ already reveals that this is so. However, as claimed in 2), the right-hand side of (\ref{eq1.1}) may not pose any positive lower bound on the product uncertainty. We will now give an explicit example which illustrates this fact. Indeed, let $n=m=2k$ for $k \in \mathbb{N}$ and consider the Gaussian state
\begin{equation}
f (x) = \left( \frac{2a}{\pi}\right)^{\frac{1}{4}} e^{- a x^2}, \hspace{1 cm} a>0
\label{eq2.8}
\end{equation}
A simple calculation then shows that
\begin{equation}
\langle \left[(\widehat x)^{2k}, (\widehat \xi)^{2k} \right] f , f \rangle_{L^2 (\mathbb{R})} = (-i )^{2k} \sqrt{\frac{2a}{\pi}} \int_{\mathbb{R}} e^{-a x^2} \left( x^{2k} \frac{d^{2k}}{dx^{2k}} - \frac{d^{2k}}{dx^{2k}} x^{2k} \right) e^{-a x^2} dx
\label{eq2.9}
\end{equation}
If we integrate the second term by parts $2k$ times, we conclude that the previous expression vanishes identically. And so, the right-hand side of (\ref{eq1.1}) does not constitute the minimum of the uncertainty in this case.

This then brings us to questions of interpretation. The state is still represented by some $f \in L^2 (\mathbb{R})$. And it is still a fact of life that $f$ and its Fourier transform $\widetilde{f}$ cannot be both sharply localized. Notice that there is nevertheless no contradiction with our analysis. From eqs.(\ref{eq2.5},\ref{eq2.6}) the dispersions are such that for instance $\Delta_A (f_s,0)$ goes to zero, and $\Delta_B (f_s,0)$ diverges as $|s| \to + \infty$, while their product becomes arbitrarily small. But this does not mean that there can be an infinite precision in the simultaneous measurement of $\widehat{A}$ and $\widehat{B}$. On the contrary, one of the two is measured with growing precision, while the other becomes coarser. So the uncertainty is still there. In particular, if the two dispersions were to become simultaneously infinitesimal, then that would imply the existence (in the limit) of a common eigenstate (albeit in a distributional sense). This is manifestly impossible.

The conclusion to be drawn from this analysis is, as we argued before, that the product of dispersions may not be a good measure of uncertainty. Also the measure of dispersion itself ($\Delta_A (f,a)$) has a drawback. Except for the linear case $\widehat{A}= \alpha \widehat{x} + \beta \widehat{\xi}$, the dispersion $\Delta_A (f,a)$ of an observable $\widehat{A} (\widehat{x},\widehat{\xi})$ is not invariant under phase-space translations $(x,\xi) \mapsto (x+x_0,\xi+\xi_0)$. To circumvent this difficulty one considers in harmonic analysis \cite{Folland2} the translation invariant dispersions $\widehat{A}(\widehat{x},\widehat{\xi})\mapsto  \widehat{A}(\widehat{x}-a,\widehat{\xi}-b)$. So, for example, if $\widehat{A}=\widehat{x}^2$, we consider the measure of dispersion:
\begin{equation}
\| ~(x-a)^2 f \|_{L^2 (\mathbb{R})}= \left(\int_{\mathbb{R}} (x-a)^4 |f(x)|^2 dx \right)^{1/2},
\label{eq2.9.A}
\end{equation}
instead of
\begin{equation}
\Delta_A (f,a^2) =\| ~(x^2-a^2) f \|_{L^2 (\mathbb{R})} = \left(\int_{\mathbb{R}} (x^2-a^2)^2 |f(x)|^2 dx \right)^{1/2}.
\label{eq2.9.B}
\end{equation}
Moreover, since the measure in (\ref{eq2.9.A}) is translation invariant, we may set, for all practical purposes, $a=0$.

All things considered, we shall consider the following measure of uncertainty:
\begin{equation}
\|\widehat{A}  f \|_{L^2 (\mathbb{R})}^2 + \|\widehat{B} f \|_{L^2 (\mathbb{R})}^2
\label{eq2.10}
\end{equation}
From the trivial inequality $\alpha^2 + \beta^2 \ge 2 \alpha \beta$, we obtain:
\begin{equation}
\|\widehat{A}  f \|_{L^2 (\mathbb{R})}^2 + \|\widehat{B}  f \|_{L^2 (\mathbb{R})}^2
 \geq 2 \|\widehat{A}  f\|_{L^2 (\mathbb{R})} ~ \|\widehat{B}  f\|_{L^2 (\mathbb{R})}
 \label{eq2.11}
\end{equation}
Equality holds if and only if
\begin{equation}
\|\widehat{A} f \|_{L^2 (\mathbb{R})} = \|\widehat{B} f\|_{L^2 (\mathbb{R})}
\label{eq2.12}
\end{equation}
It is a well known fact in harmonic analysis \cite{Folland2} that, for instance, the Heisenberg uncertainty principle:
\begin{equation}
\|xf\|_{L^2 (\mathbb{R})} ~ \|\xi \widehat{f}\|_{L^2 (\mathbb{R})} \geq C \|f\|_{L^2 (\mathbb{R})}^2,
\label{eqCorrectionsB1}
\end{equation}
is equivalent to the inequality
\begin{equation}
\|xf\|_{L^2 (\mathbb{R})}^2 + \|\xi \widehat{f}\|_{L^2 (\mathbb{R})}^2 \geq K \|f\|_{L^2 (\mathbb{R})}^2,
\label{eqCorrectionsB2}
\end{equation}
for some constant $K >0$ and $C=\frac{1}{2}$. We already know from (\ref{eq2.11}) that (\ref{eqCorrectionsB1}) implies (\ref{eqCorrectionsB2}). To show that the converse is also true, we consider again the scale transformation (\ref{eq2.1}) $f_s=D_s f$. If (\ref{eqCorrectionsB2}) holds for all $f \in L^2 (\mathbb{R})$, then it also holds for $f_s$. From (\ref{eq2.3}) and (\ref{eqCorrectionsB2}), we obtain:
\begin{equation}
s^2\|xf\|_{L^2 (\mathbb{R})}^2 +s^{-2} \|\xi \widehat{f}\|_{L^2 (\mathbb{R})}^2 \geq K \|f\|_{L^2 (\mathbb{R})}^2,
\label{eqCorrectionsB3}
\end{equation}
which holds for all $f \in L^2 (\mathbb{R})$ and all $s \ne 0$. Taking the infimum on the left-hand side with respect to $s$, we recover (\ref{eqCorrectionsB1}).

Inequality (\ref{eq2.11}) shows that if there is a minimum of product of dispersions (e.g. $\widehat{A}=(\widehat{x})^n$ and $\widehat{B}=(\widehat{\xi})^n$), then (\ref{eq2.10}) will also be bounded from below. Conversely, if there is no lower bound on the product of dispersions, then that does not preclude a lower bound on (\ref{eq2.10}). Notice that from (\ref{eq2.5},\ref{eq2.6}) if $|s| \to + \infty$, then the product of dispersions vanishes while (\ref{eq2.10}) does not.

So in the sequel, we shall consider the expression  (\ref{eq2.10}) as our measure of uncertainty rather than the product of dispersions $\Delta_A (f,a) \Delta_B (f,b)$. Thus, {\it minimal uncertainty} states will mean, for all practical purposes, states which minimize the uncertainty measures of the form (\ref{eq2.10}).

This will help clarify whether there are coherent states for this algebra. Let us briefly explain what we have in mind. For the usual Heisenberg-Weyl algebra, coherent states can be constructed from any of the following three definitions \cite{Dey4}:

\vspace{0.3 cm}
\noindent
(i) as eigenstates of the annihilation operators $\widehat{a}_j=\frac{1}{\sqrt{2}} \left(\widehat{x}_j+ i \widehat{\xi}_j \right)$,

\vspace{0.2 cm}
\noindent
(ii) by applying Glauber's displacement operator $\widehat{D}( \alpha)=\exp \left(\alpha\cdot \widehat{a}^{\dagger} - \overline{\alpha} \cdot \widehat{a} \right)$ to the vacuum state, and

\vspace{0.2 cm}
\noindent
(iii) as quantum states that minimize the uncertainty relation $\Delta_{x_j}~ \Delta_{\xi_j} = \frac{1}{2}$ for all $j$, with equal uncertainties in each coordinate $\Delta_{x_j}= \Delta_{\xi_j}$.

\vspace{0.3 cm}
\noindent
In general, coherent states may fail to satisfy the three conditions at all times, see for instance \cite{Dey1}, where the first two conditions are satisfied, but the third one is not. If all three conditions are respected, then the states are called {\it intelligent coherent states}.

In this work, we shall address the third condition. However, we will see that when $\widehat A$ and $\widehat B$ are fundamental observables of the non-canonical algebra their product of dispersions is not bounded from below by a positive constant, while the uncertainty (\ref{eq2.10}) satisfies an inequality of the form (\ref{eqIntroduction1}) for some positive constant $C$. So, it is not true that the dispersions of $\widehat A$ and $\widehat B$ can be simultaneously equal to zero, but an uncertainty principle using the product of dispersions is unable to capture this property. 

Hence, as pointed out previously, our measure of uncertainty will be (\ref{eq2.10}) rather than the product of dispersions. Accordingly, coherent states are defined as the states that minimize simultaneously all the uncertainties (\ref{eq2.10}), where $\widehat A$ and $\widehat B$ are noncommuting fundamental variables in the new algebra. We will discuss the existence of such states in section 6.3.

\section{Non-canonical extension of the Heisenberg-Weyl algebra}

Given the physical motivations stated in the introduction, we shall now consider several aspects of the non-canonical phase-space noncommutative algebra of Refs.\cite{Bastos7,Bastos8}.

We consider a two-dimensional configuration space with noncommuting coordinates $\widehat{q}=(\widehat{q}_1, \widehat{q}_2)$ and canonical conjugate momenta $\widehat{p}=(\widehat{p}_1, \widehat{p}_2)$. In \cite{Bastos7,Bastos8}, $\widehat{q},\widehat{p}$ are not interpreted as the position and momentum of some particle, but rather as the scale factors appearing in the Kantowski-Sachs metric and their conjugate momenta. Other applications of such algebras are also of interest (see e.g. \cite{Delduc}). The non-canonical algebra reads:
\begin{equation}
\begin{array}{l}
\left[ \widehat{q}_1,  \widehat{q}_2 \right] = i \theta( \widehat{I} + \theta  \widehat{R})\\
\left[ \widehat{p}_1,  \widehat{p}_2 \right] = i \left( \eta \widehat{I} + (1+ \sqrt{1 - \xi})^2 \widehat{R} \right)\\
\left[ \widehat{q}_1,  \widehat{p}_1 \right] = \left[ \widehat{q}_2,  \widehat{p}_2 \right] =i \left( \widehat{I} + \theta (1+ \sqrt{1 - \xi}) \widehat{R} \right)
\end{array}
\label{eq3.1}
\end{equation}
while all the remaining commutators vanish. Here $ \widehat{R}$ denotes the operator
\begin{equation}
 \widehat{R} = \epsilon \left(  \widehat{q}_1 + \frac{\theta}{1+ \sqrt{1- \xi}}  \widehat{p}_2 \right)
\label{eq3.2}
\end{equation}
 Also, $\theta, \eta, \epsilon$ are positive constants, and $\xi= \theta \eta <1$. The constants $\theta$ and $\eta$ measure the noncommutativity in the configuration and momentum sectors, respectively. Indeed, if $\epsilon =0$, $\widehat{R}$ vanishes and one recovers the canonical phase-space noncommutative algebra \cite{Bastos1,Bastos2}:
\begin{equation}
\left[ \widehat{q}_1,  \widehat{q}_2 \right] = i \theta \widehat{I}, \hspace{0.5 cm} \left[ \widehat{p}_1,  \widehat{p}_2 \right] = i \eta \widehat{I}, \hspace{0.5 cm} \left[ \widehat{q}_1,  \widehat{p}_1 \right] = \left[ \widehat{q}_2,  \widehat{p}_2 \right] = i  \widehat{I}
\label{eq3.3}
\end{equation}
On the other hand $\epsilon$ is responsible for the non-canonical character of the algebra. Even if $\theta=\eta=0$, one still obtains a non-canonical noncommutative deformation of the HW algebra:
\begin{equation}
\left[ \widehat{q}_1,  \widehat{q}_2 \right] = 0, \hspace{0.5 cm} \left[ \widehat{p}_1,  \widehat{p}_2 \right] = 4i \epsilon \widehat{q}_1, \hspace{0.5 cm} \left[ \widehat{q}_1,  \widehat{p}_1 \right] = \left[ \widehat{q}_2,  \widehat{p}_2 \right] = i  \widehat{I}
\label{eq3.4}
\end{equation}
Notice that this algebra can be interpreted as an effective algebra for a system coupled to an external non-constant magnetic field (see \cite{Delduc} for details).

The consistency of the algebra (\ref{eq3.1}) is easily established. Indeed the Jacobi identity is a straightforward consequence of the fact that the algebra is equivalent to the HW algebra. Indeed, the following map is a nonlinear symplectomorphism to the HW algebra\footnote{Here we use the {\it classical} notion of symplectomorphism as a bijection $\phi: E \to V$ from a symplectic space $(E, \sigma)$ to another symplectic space $(V,\omega)$ such that $\phi^{\ast} \omega = \sigma$, i.e. $\omega\left( \phi (z), \phi (z^{\prime})\right) = \sigma (z, z^{\prime})$ for all $z, z^{\prime} \in E$.}:
\begin{equation}
\begin{array}{l l l}
\widehat{q}_1 = \lambda \widehat{x}_1 - \frac{\theta}{2 \lambda} \widehat{\xi}_2 + E \widehat{x}_1^2 & \hspace{0.5 cm} & \widehat{q}_2 = \lambda \widehat{x}_2 + \frac{\theta}{2 \lambda} \widehat{\xi}_1\\
& & \\
\widehat{p}_1 = \mu \widehat{\xi}_1 + \frac{\eta}{2 \mu} \widehat{x}_2 & \hspace{0.5 cm} & \widehat{p}_2  = \mu \widehat{\xi}_2 - \frac{\eta}{2 \mu} \widehat{x}_1 + F \widehat{x}_1^2
\end{array}
\label{eq3.5}
\end{equation}
Here $\mu, \lambda$ are real parameters such that $2 \mu \lambda = 1 + \sqrt{1 - \xi}$, and
\begin{equation}
E= - \frac{\theta F}{1+\sqrt{1- \xi}}, \hspace{1 cm} F= - \frac{\lambda}{\mu} \epsilon  \sqrt{1- \xi} (1+\sqrt{1- \xi})
\label{eq3.6}
\end{equation}
The inverse transformation is easily established:
\begin{equation}
\begin{array}{l l l}
\widehat{x}_1 = \frac{1}{\sqrt{1- \xi}} \left(\mu \widehat{q}_1  + \frac{\theta}{2 \lambda} \widehat{p}_2 \right) & \hspace{0.5 cm} & \widehat{x}_2 = \frac{1}{\sqrt{1- \xi}} \left(\mu \widehat{q}_2  - \frac{\theta}{2 \lambda} \widehat{p}_1 \right) \\
& & \\
\widehat{\xi}_1 = \frac{1}{\sqrt{1- \xi}} \left(\lambda \widehat{p}_1  - \frac{\eta}{2 \mu} \widehat{q}_2 \right) & \hspace{0.5 cm} & \widehat{\xi}_2 = \frac{1}{\sqrt{1- \xi}} \left(\lambda \widehat{p}_2  + \frac{\eta}{2 \mu} \widehat{q}_1 - \frac{F \mu}{ \epsilon^2 \sqrt{1- \xi}}  \widehat{R}^2 \right)
\end{array}
\label{eq3.7}
\end{equation}
The variables $( \widehat{x}_1,  \widehat{x}_2,  \widehat{\xi}_1,  \widehat{\xi}_2)$  satisfy the HW algebra:
\begin{equation}
\left[  \widehat{x}_1,  \widehat{x}_2 \right]= \left[  \widehat{\xi}_1,  \widehat{\xi}_2 \right]=0, \hspace{1 cm} \left[  \widehat{x}_1,  \widehat{\xi}_1 \right]= \left[  \widehat{x}_2,  \widehat{\xi}_2 \right]= i  \widehat{I}
\label{eq3.8}
\end{equation}
Notice that the symplectomorphism is not unique. Indeed, the composition of the symplectomorphism with an arbitrary unitary transformation yields an equally valid symplectomorphism. However, all physical predictions (expectation values, probabilities, eigenvalues) are invariant under a choice of symplectomorphism \cite{Bastos2,Bastos3}, so we may safely choose (\ref{eq3.5},\ref{eq3.7}) for the remainder of this work.

From this map, we can thus obtain a differential representation of the algebra in $L^2 (\mathbb{R}^2)$:
\begin{equation}
\begin{array}{l}
(\widehat{q}_1 f) (x_1, x_2) = \left(\lambda x_1 + \frac{i \theta}{2 \lambda} \frac{\partial}{\partial x_2} + E x_1^2 \right) f (x_1, x_2) \\
\\
(\widehat{q}_2 f) (x_1, x_2) = \left(\lambda x_2 - \frac{i \theta}{2 \lambda} \frac{\partial}{\partial x_1}  \right) f (x_1, x_2)\\
\\
(\widehat{p}_1 f) (x_1, x_2) = \left(- i \mu \frac{\partial}{\partial x_1} +  \frac{ \eta}{2 \mu}  x_2 \right) f (x_1, x_2)\\
\\
(\widehat{p}_2 f) (x_1, x_2) = \left(- i \mu \frac{\partial}{\partial x_2} - \frac{ \eta}{2 \mu}  x_1 + F x_1^2 \right) f (x_1, x_2)
\end{array}
\label{eq3.9}
\end{equation}
and the corresponding maximal domains
\begin{equation}
\begin{array}{l}
\text{Dom} (\widehat{q}_1) := \left\{f \in L^2(\mathbb{R}^2): ~\left( x_1 + \frac{i \theta}{2 \lambda^2} \frac{\partial}{\partial x_2} + \frac{E}{\lambda} x_1^2\right) f (x_1,x_2) \in L^2(\mathbb{R}^2) \right\}\\
\\
 \text{Dom}(\widehat{q}_2) := \left\{f \in L^2(\mathbb{R}^2): ~\left( x_2 - \frac{i \theta}{2 \lambda^2} \frac{\partial}{\partial x_1}\right) f (x_1,x_2) \in L^2(\mathbb{R}^2) \right\}\\
\\
 \text{Dom}(\widehat{p}_1) := \left\{f \in L^2(\mathbb{R}^2): ~\left( x_2 - \frac{2i \mu^2}{\eta} \frac{\partial}{\partial x_1} \right) f (x_1,x_2) \in L^2(\mathbb{R}^2) \right\}\\
\\
\text{Dom} (\widehat{p}_2):= \left\{f \in L^2(\mathbb{R}^2): ~\left( x_1 + \frac{2i \mu^2}{\eta} \frac{\partial}{\partial x_2} - \frac{2 \mu F}{\eta} x_1^2\right) f (x_1,x_2) \in L^2(\mathbb{R}^2) \right\}
\end{array}
\label{eq3.10}
\end{equation}

We leave to the reader the proof of the following lemma.

\begin{Lemma}\label{lemma3.0}
The operators $\widehat{q_1}, \widehat{q_2}, \widehat{p_1}, \widehat{p_2}$ are self-adjoint on their maximal domains.
\end{Lemma}

This result deserves some comments. In some theories of noncommutative quantum mechanics (see e.g. \cite{Dey1,Dey2,Dey3,Dey4}) the fundamental operators may not be self-adjoint or not even hermitian (for instance in PT symmetric systems). An example would be a representation of the $q$-deformed oscillator algebra on the unit circle acting on Rogers-Sz\"ego polynomials \cite{Dey1}. The self-adjoint representation (\ref{eq3.9}) on the maximal domains (\ref{eq3.10}) comes from the fact that our noncanonical algebra is globally isomorphic with the Heisenberg-Weyl algebra and the isomorphism is a polynomial of degree at most 2. Notice that the nature of our deformation of the Heisenberg-Weyl algebra is somewhat different from the $q$-deformation of \cite{Dey1}. Here the commutation relation of $\widehat{q}_1     $ and $\widehat{p}_1$ is deformed, but the operators $\widehat{q}_1 $ and $\widehat{q}_2$ no longer commute. In \cite{Dey1}, a q-deformation is performed on each oscillator separately. The q-deformation may be more restrictive than our deformation, as it leads to a minimal length and a minimal momentum, while ours does not. 

\section{Stability of the algebra}

Before we proceed, let us analyze the stability of our algebra.

Consider some Lie algebra $\mathcal{A}_0$ with product $\left[ \cdot, \cdot \right]_0$ defined on a vector space $V$ over a field $\mathbb{K}$. A formal deformation of $\mathcal{A}_0$ is an algebra $\mathcal{A}_{\rho}$ on the space $V \otimes \mathbb{K} \left[ \rho \right]$ (where $\mathbb{K} \left[ \rho \right]$ is the ring of formal power series), defined by:
\begin{equation}
\left[A, B \right]_{\rho} =\left[A, B \right]_0+ \sum_{k=1}^{\infty} B_k (A,B) \rho^k,
\label{eq4.1}
\end{equation}
where $A,B,B_k(A,B) \in V$ $(k \ge 1)$ and $\rho \in \mathbb{K}$. In this instance, one has instead a 3 parameter deformation $(\theta, \eta, \epsilon)$, but the essential arguments are not substantially altered, so we will keep to the simpler one-parameter deformation.

If all deformations $\mathcal{A}_{\rho}$ are isomorphic to $\mathcal{A}_0$, then $\mathcal{A}_0$ is said to be stable or rigid. This concept is paramount in the so-called stable model approach to model building \cite{Faddeev,Flato,Smale}. From this point of view, one aims to construct models with properties which remain stable under small changes of the parameters. If one has, for instance, an unstable algebra, one deforms it until one obtains a stable algebra. It is well known that the passage from non-relativistic to relativistic or from classical to quantum mechanics, can be interpreted as the transition from unstable to stable theories \cite{Vilela}.

A simple inspection of (\ref{eq4.1}) reveals that the maps $B_k $ must be 2-cochains in $V$. The imposition of the Jacobi indentity entails that $B_1$ be a 2-cocycle. The Rigidity Theorem of Nijenhuis and Richardson \cite{Nijenhuis1,Nijenhuis2} states that if the second co-homology group of the algebra $\mathcal{A}_0$ is trivial, so that $B_1$ is in fact a 2-coboundary, then $\mathcal{A}_0$ is stable. This is what happens if $\mathcal{A}_0$ is a semi-simple Lie algebra \cite{Jacobson}. However, this is a sufficient condition, but not a necessary one. A non-trivial second co-homology group may not be an obstruction to stable deformations. This is an important point for our purposes. Indeed, the HW algebra has a non-trivial second co-homology group. Nevertheless it is a stable algebra. The reason is that there exists a non-linear isomorphism to a stable algebra \cite{Vilela}. Since our algebra is also non-linearly isomorphic with the HW algebra, one concludes that our algebra is stable.

\section{Functional spaces}

\subsection{Functional spaces for the uncertainty principle}

For our purposes it will prove useful to consider the following functional spaces.

\begin{Definition}\label{def3.1}
Let $\mathcal{N}= \left\{(\widehat{q}_1, \widehat{q}_2), (\widehat{p}_1, \widehat{p}_2), (\widehat{q}_1, \widehat{p}_1), (\widehat{q}_2, \widehat{p}_2) \right\}$. For each $\alpha = (\widehat{u}, \widehat{v}) \in \mathcal{N}$, we define
\begin{equation}
\mathcal{B}^{\alpha} (\mathbb{R}^2) := \left\{f \in \mathcal{S}^{\prime} (\mathbb{R}^2): ~ \|f \|_{\alpha} < + \infty \right\}
\label{eq3.22}
\end{equation}
where
\begin{equation}
\|f \|_{\alpha}^2 := 2 \|f \|_{L^2 (\mathbb{R}^2)} + \|\widehat{u} f\|_{L^2 (\mathbb{R}^2)}^2 + \|\widehat{v} f\|_{L^2 (\mathbb{R}^2)}^2
\label{eq3.23}
\end{equation}
We shall also consider the space
\begin{equation}
\mathcal{B} (\mathbb{R}^2) := \left\{f \in \mathcal{S}^{\prime} (\mathbb{R}^2): ~ \|f \|_{\mathcal{B}} < + \infty \right\}
\label{eq3.23.1}
\end{equation}
where
\begin{equation}
\begin{array}{c}
\|f \|_{\mathcal{B}}^2 := 2 \|f\|_{L^2 (\mathbb{R}^2)}^2 + \| ~ (x_1+ E x_1^2 / \lambda) f \|_{L^2 (\mathbb{R}^2)}^2 + \|x_2 f \|_{L^2 (\mathbb{R}^2)}^2 + \\
\\
+  \| \xi_1 \widetilde{f} \|_{L^2 (\mathbb{R}^2)}^2 + \| \xi_2 \widetilde{f} \|_{L^2 (\mathbb{R}^2)}^2 = \\
\\
=  \int_{\mathbb{R}^2} \left( 1 + \left(x_1 + E x_1^2 / \lambda \right)^2 + x_2^2 \right) |f(x)|^2 dx + \\
\\
+ \int_{\mathbb{R}^2} \left( 1 + \xi_1^2 + \xi_2^2 \right) | \widetilde{f} (\xi)|^2 d \xi.
\end{array}
\label{eq3.23.2}
\end{equation}

\end{Definition}

We leave to the reader the simple task of verifying that the spaces $\mathcal{B}^{\alpha} (\mathbb{R}^2)$ and $\mathcal{B} (\mathbb{R}^2)$ are complex normed vector spaces. In fact they are Hilbert spaces:

\begin{Proposition}\label{prop3.3}
The spaces $\mathcal{B}^{\alpha} (\mathbb{R}^2)$ for $\alpha = (\widehat{u}, \widehat{v})\in \mathcal{N}$ and $\mathcal{B} (\mathbb{R}^2)$ endowed with the inner products
\begin{equation}
\langle f,g \rangle_{\alpha} : = 2 \langle f,g \rangle_{L^2 (\mathbb{R}^2)} + \langle \widehat{u}f,\widehat{u}g \rangle_{L^2 (\mathbb{R}^2)} + \langle \widehat{v}f,\widehat{v}g \rangle_{L^2 (\mathbb{R}^2)}
\label{eq3.24}
\end{equation}
and
\begin{equation}
\begin{array}{c}
\langle f,g \rangle_{\mathcal{B}} := 2 \langle f,g \rangle_{L^2 (\mathbb{R}^2)} + \langle \left(x_1 + E x_1^2 / \lambda \right) f , \left(x_1 + E x_1 / \lambda \right) g \rangle_{L^2 (\mathbb{R}^2)} + \\
\\
+ \langle x_2 f, x_2 g \rangle_{L^2 (\mathbb{R}^2)} + \langle \xi_1 \widetilde{f}, \xi_1 \widetilde{g} \rangle_{L^2 (\mathbb{R}^2)} + \langle \xi_2 \widetilde{f}, \xi_2 \widetilde{g} \rangle_{L^2 (\mathbb{R}^2)} =\\
\\
=  \int_{\mathbb{R}^2} \left( 1 + \left(x_1 + E x_1^2 / \lambda \right)^2 + x_2^2 \right) f (x) \overline{g(x)} dx + \\
\\
+ \int_{\mathbb{R}^2} \left( 1 + \xi_1^2 + \xi_2^2 \right)  \widetilde{f} (\xi)  \overline{\widetilde{g} (\xi)} d \xi,
\end{array}
\label{eq3.24.1}
\end{equation}
respectively, are Hilbert spaces.
\end{Proposition}

\begin{proof}
By the Cauchy-Schwarz inequality, we have
\begin{equation}
\begin{array}{c}
\left|\langle f,g \rangle_{\alpha}  \right| \leq 2 \|f \|_{L^2 (\mathbb{R}^2)} \|g \|_{L^2 (\mathbb{R}^2)} + \|\widehat{u}f\|_{L^2 (\mathbb{R}^2)} \|\widehat{u}g \|_{L^2 (\mathbb{R}^2)} + \\
\\
+ \|\widehat{v}f \|_{L^2 (\mathbb{R}^2)} \|\widehat{v}g \|_{L^2 (\mathbb{R}^2)} \leq 4 \|f\|_{\alpha} \|g \|_{\alpha}
\end{array}
\label{eq3.25}
\end{equation}
which shows that $\langle \cdot, \cdot \rangle_{\alpha} $ is a well defined operation $\mathcal{B}^{\alpha} (\mathbb{R}^2) \times \mathcal{B}^{\alpha} (\mathbb{R}^2) \to \mathbb{C}$. Since $\langle f,f \rangle_{\alpha} = \|f\|_{\alpha}^2$ it is straightforward to prove that $\langle \cdot, \cdot \rangle_{\alpha} $ is an inner product. And so $\mathcal{B}^{\alpha} (\mathbb{R}^2) = \text{Dom}(\widehat{u}) \cap \text{Dom}(\widehat{v})$ is a pre-Hilbert space. It remains to prove completeness. We shall prove the case $\alpha = (\widehat{q_1}, \widehat{q_2})$. The remaining cases are proved in a similar fashion. The proof follows the standard procedure for Sobolev spaces \cite{Evans,Leoni}.

Let then $(f_n)_{n \in \mathbb{N}}$ be a Cauchy sequence in $\mathcal{B}^{\alpha} (\mathbb{R}^2)$ with $\alpha = (\widehat{q_1}, \widehat{q_2})$. Then for any $\epsilon >0$, there exists $N \in \mathbb{N}$ such that
\begin{equation}
\|f_n-f_m\|_{\alpha} = \left( 2\|f_n-f_m\|_{L^2 (\mathbb{R}^2)}^2 + \|\widehat{q_1}(f_n-f_m)\|_{L^2 (\mathbb{R}^2)}^2 + \|\widehat{q_2}(f_n-f_m)\|_{L^2 (\mathbb{R}^2)}^2 \right)^{1/2} < \epsilon
\label{eq3.25.1}
\end{equation}
for all $n,m \in \mathbb{N}$ with $n,m \geq N$.

It follows that $(f_n)_{n \in \mathbb{N}}$, $(\widehat{q_1} f_n)_{n \in \mathbb{N}}$ and $(\widehat{q_2} f_n)_{n \in \mathbb{N}}$ are Cauchy sequences in $L^2 (\mathbb{R}^2)$. Since $L^2 (\mathbb{R}^2)$ is complete, there exist $f,g,h \in L^2 (\mathbb{R}^2)$ such that
\begin{equation}
\|f_n-f\|_{L^2 (\mathbb{R}^2)} \to 0, \hspace{1 cm} \|\widehat{q_1}f_n- g\|_{L^2 (\mathbb{R}^2)} \to 0, \hspace{1 cm} \mbox{and } \|\widehat{q_2}f_n- h \|_{L^2 (\mathbb{R}^2)} \to 0
\label{eq3.25.2}
\end{equation}
as $n \to \infty$. The proof is completed, provided we prove that $g = \widehat{q_1} f$ and $h = \widehat{q_2} f$ a.e.. We prove the first identity and leave the second one to the reader. Let $t \in \mathcal{S}(\mathbb{R}^2)$. Then we have by the continuity of the inner product and the fact that $\widehat{q_1}$ is self-adjoint:
\begin{equation}
\begin{array}{c}
\langle  f,\widehat{q_1} t \rangle_{L^2 (\mathbb{R}^2)} = \langle \lim f_n,\widehat{q_1} t \rangle_{L^2 (\mathbb{R}^2)} = \lim \langle f_n,\widehat{q_1} t \rangle_{L^2 (\mathbb{R}^2)} =\\
 \\
 = \lim \langle \widehat{q_1}f_n, t \rangle_{L^2 (\mathbb{R}^2)} =\langle \lim \widehat{q_1}f_n, t \rangle_{L^2 (\mathbb{R}^2)} =\langle g, t \rangle_{L^2 (\mathbb{R}^2)}
\end{array}
\label{eq3.25.3}
\end{equation}
which holds for any $t \in \mathcal{S}(\mathbb{R}^2)$. Since  $\mathcal{S}(\mathbb{R}^2)$ is dense in $L^2 (\mathbb{R}^2)$ and $\widehat{q_1}$ is self-adjoint, we conclude that $g = \widehat{q_1}^{\ast}f = \widehat{q_1} f$ a.e..

The fact that $\mathcal{B} (\mathbb{R}^2)$ is also a pre-Hilbert space is proved in a similar fashion. On the other hand, completeness of $\mathcal{B} (\mathbb{R}^2)$ is then a simple consequence of Proposition \ref{prop3.3.1} (see below).
\end{proof}

The next proposition reveals that all the spaces $\mathcal{B}^{\alpha} (\mathbb{R}^2)$ for $\alpha = (\widehat{u}, \widehat{v})\in \mathcal{N}$ and $\mathcal{B} (\mathbb{R}^2)$ coincide.
\begin{Proposition}\label{prop3.3.1}
For any $\alpha = (\widehat{u}, \widehat{v})\in \mathcal{N}$ and all $f \in \mathcal{S}^{\prime} (\mathbb{R}^2)$, we have
\begin{equation}
\|f \|_{\alpha} \asymp \|f \|_{\mathcal{B}}.
\label{eq3.25.4}
\end{equation}
\end{Proposition}

\begin{proof}
We will prove the result for $\alpha =(\widehat{q_1},\widehat{q_2})$. The remaining cases are proved in a similar fashion. By the triangle inequality, and the fact that $(|a|+|b|)^2 \le 2 |a|^2 + 2 |b|^2$, we have
\begin{equation}
\begin{array}{c}
\|f \|_{\alpha}^2 = 2 \|f \|_2^2 + \|\left( \lambda \widehat{x_1}- \frac{\theta}{2 \lambda} \widehat{\xi_2} + E \widehat{x_1}^2 \right) f \|_{L^2 (\mathbb{R}^2)}^2 + \| \left(\lambda \widehat{x_2} + \frac{\theta}{2 \lambda}\widehat{\xi_1} \right)f \|_{L^2 (\mathbb{R}^2)}^2 \leq \\
\\
\leq 2 \|f\|_{L^2 (\mathbb{R}^2)}^2 + \left(| \lambda| ~ \|(x_1 + E x_1^2 / \lambda) f \|_{L^2 (\mathbb{R}^2)} + \left| \frac{\theta}{2 \lambda}\right| ~ \| \xi_2 \widetilde{f}\|_{L^2 (\mathbb{R}^2)} \right)^2+\\
 \\
 + \left(| \lambda| ~ \|x_2  f \|_{L^2 (\mathbb{R}^2)} + \left| \frac{\theta}{2 \lambda}\right| ~ \| \xi_1 \widetilde{f}\|_{L^2 (\mathbb{R}^2)} \right)^2 \leq \\
\\
\leq 2 \|f\|_{L^2 (\mathbb{R}^2)}^2 +2  \lambda^2 ~ \|(x_1 + E x_1^2 / \lambda) f \|_{L^2 (\mathbb{R}^2)}^2 + \frac{\theta^2}{2 \lambda^2} \| \xi_2 \widetilde{f} \|_{L^2 (\mathbb{R}^2)}^2 + \\
\\
+2 \lambda^2 \|x_2  f \|_{L^2 (\mathbb{R}^2)}^2+ \frac{\theta^2}{2 \lambda^2} \| \xi_1 \widetilde{f}\|_{L^2 (\mathbb{R}^2)}^2 \leq C \|f \|_{\mathcal{B}}^2,
\end{array}
\label{eq3.25.5}
\end{equation}
where $C = \max \left\{ 2 \lambda^2 , \frac{\theta^2}{2 \lambda^2}, 1 \right\}$. This shows that $\|f \|_{\alpha} \lesssim \|f \|_{\mathcal{B}}$.

We next prove the converse result. We start by assuming that $f$ is real. It follows that:
\begin{equation}
\begin{array}{c}
\|\widehat{q_1}f \|_{L^2 (\mathbb{R}^2)}^2 = \int_{\mathbb{R}^2} \left| \lambda x_1 f(x) + \frac{i \theta}{2 \lambda} \frac{\partial f(x)}{\partial x_2} + E x_1^2  f(x)  \right|^2 dx= \\
\\
=  \lambda^2 \int_{\mathbb{R}^2} \left(  x_1 + \frac{E}{\lambda} x_1^2  \right)^2 ~ | f(x)|^2 dx  +  \frac{\theta^2}{4 \lambda^2}  \int_{\mathbb{R}^2} \left| \frac{\partial f(x)}{\partial x_2}\right|^2 dx =\\
\\
= \lambda^2 \| (x_1 + E x_1^2 / \lambda) f\|_{L^2 (\mathbb{R}^2)}^2 + \frac{\theta^2}{4 \lambda^2}  \| \xi_2 \widetilde{f}\|_{L^2 (\mathbb{R}^2)}^2.
\end{array}
\label{eq3.25.6}
\end{equation}
Similarly,
\begin{equation}
\|\widehat{q_2}f\|_{L^2 (\mathbb{R}^2)}^2 = \lambda^2 \|x_2 f\|_{L^2 (\mathbb{R}^2)}^2 +  \frac{\theta^2}{4 \lambda^2}  \| \xi_1 \widetilde{f}\|_{L^2 (\mathbb{R}^2)}^2.
\label{eq3.25.7}
\end{equation}
Altogether, we obtain
\begin{equation}
\|f\|_{\alpha}^2 \geq K^2 \|f\|_{\mathcal{B}}^2,
\label{eq3.25.8}
\end{equation}
where $K = min \left\{ | \lambda|, \left|\frac{\theta}{2 \lambda} \right| , 1\right\}$.

We next prove that the previous inequality is also valid, even if $f$ is not real. Let us write $f=f_R + i f_I$, where $f_R$ and $f_I$ are both real. Notice that $\| \overline{f}\|_{\alpha} = \|f\|_{\alpha}$ and the same is valid for $\| \cdot\|_{\mathcal{B}}$. Now assume that for some $f=f_R + i f_I$, (\ref{eq3.25.8}) does not hold so that
\begin{equation}
\|f\|_{\alpha} < K \|f\|_{\mathcal{B}}.
\label{eq3.25.9}
\end{equation}
It then follows from (\ref{eq3.25.8},\ref{eq3.25.9}) and the Parallelogram Law that
\begin{equation}
\begin{array}{c}
2 K^2 \|f\|_{\mathcal{B}}^2 > 2 \|f\|_{\alpha}^2 = \|f_R + i f_I\|_{\alpha}^2+ \|f_R -i f_I\|_{\alpha}^2 = 2 \|f_R \|_{\alpha}^2 + 2\|f_I\|_{\alpha}^2 \geq \\
\\
\geq K^2 \left(  2 \|f_R \|_{\mathcal{B}}^2 + 2\|f_I\|_{\mathcal{B}}^2 \right) = K^2 \left(\|f_R + i f_I\|_{\mathcal{B}}^2+ \|f_R -i f_I\|_{\mathcal{B}}^2  \right)=\\
\\
=2 K^2 \|f\|_{\mathcal{B}}^2
\end{array}
\label{eq3.26.10}
\end{equation}
and we have a contradiction.
\end{proof}

\begin{Lemma}\label{lemma3.3.1}
The spaces $\mathcal{B}^{\alpha} (\mathbb{R}^2)$ are dense subsets of $L^2 (\mathbb{R}^2)$ for all $\alpha \in \mathcal{N}$.
\end{Lemma}

\begin{proof}
This is a trivial consequence of the facts that $\mathcal{S} (\mathbb{R}^2) \subset \mathcal{B}^{\alpha} (\mathbb{R}^2) \subset L^2 (\mathbb{R}^2) $ and that
$\mathcal{S} (\mathbb{R}^2)$ is dense in $L^2(\mathbb{R}^2)$.
\end{proof}

\begin{Proposition}\label{prop3.4}
Let $\alpha = (\widehat{u},\widehat{v}) \in \mathcal{N}$. Then $\widehat{u}$ and $\widehat{v}$ are bounded linear operators from $\mathcal{B}^{\alpha} (\mathbb{R}^2) \to L^2 (\mathbb{R}^2)$.
\end{Proposition}

\begin{proof}
This follows straightforwardly from
\begin{equation}
\|\widehat{u} f \|_{L^2 (\mathbb{R}^2)} \leq \left( 2 \| f \|_{L^2 (\mathbb{R}^2)}^2 + \|\widehat{u} f \|_{L^2 (\mathbb{R}^2)}^2 + \|\widehat{v} f\|_{L^2 (\mathbb{R}^2)}^2\right)^{\frac{1}{2}} = \| f \|_{\alpha}
\label{eq3.26}
\end{equation}
and the same for $\widehat{v}$. In particular, we have $\|\widehat{u}\|_{op} \leq 1$ and $\|\widehat{v} \|_{op} \leq 1$.
\end{proof}

\section{Uncertainty principle for the algebra}

Let us apply the inequality (\ref{eq1.1}) to our algebra acting on an appropriate subset of $\mathcal{B} (\mathbb{R}^2)$ where all quantities are well defined. We then get for a normalized $f$:
\begin{equation}
\begin{array}{l}
\Delta_{q_1} (f,a_1) \Delta_{q_2} (f,a_2) \geq \frac{\theta}{2} \left| 1+ \theta \langle f, \widehat{R} f \rangle_{L^2 (\mathbb{R}^2)} \right|\\
\\
\Delta_{q_1} (f,a_1) \Delta_{p_1} (f,b_1) \geq \frac{1}{2} \left| 1+ \theta (1+ \sqrt{1- \xi}) \langle f, \widehat{R} f  \rangle_{L^2 (\mathbb{R}^2)} \right|\\
\\
\Delta_{q_2} (f,a_2) \Delta_{p_2} (f,b_2) \geq \frac{1}{2} \left| 1+ \theta (1+ \sqrt{1- \xi}) \langle f, \widehat{R} f  \rangle_{L^2 (\mathbb{R}^2)} \right|\\
\\
\Delta_{p_1} (f,b_1) \Delta_{p_2} (f,b_2) \geq \frac{1}{2} \left| \eta + \left(1+   \sqrt{1- \xi} \right)^2 \langle f, \widehat{R} f  \rangle_{L^2 (\mathbb{R}^2)} \right|
\end{array}
\label{eq6.1}
\end{equation}
Using a suitable phase-space translation, one can choose the expectation value of $\widehat{R}$, so that the right-hand sides in (\ref{eq6.1}) vanish separately. So these inequalities pose no minimum bound on the product of the dispersions. This is in agreement with our analysis in section 2.

We will use the uncertainty measure defined in (\ref{eq2.10}) as our measure of uncertainty for the following reasons:

\vspace{0.2 cm}
\begin{enumerate}
\item As explained in section 2, it is a more useful measure of uncertainty, when compared with the product of dispersions, if one has composite operators of the fundamental Heisenberg-Weyl position and momentum observables. This is precisely the case at hand. Indeed, from the isomorphism (\ref{eq3.5}), we conclude that the operators $\widehat{q}_1, ~\widehat{q}_2, ~\widehat{p}_1,~\widehat{p}_2$ are composite operators of the fundamental Heisenberg-Weyl position and momentum operators $\widehat{x}_1, ~\widehat{x}_2, ~\widehat{\xi}_1, ~\widehat{\xi}_2$, for which the products of dispersions can be made arbitrarily small.

\item We will also see that this measure facilitates the calculation of the minima in comparison with measures of the Cowling-Price type, which lead to highly non-linear Euler-Lagrange equations.

\item Entropic uncertainty measures such as (\ref{eqCorrectionsC3}) are difficult to use in this setting. Recall that, for the ordinary Heisenberg-Weyl algebra $\left[\widehat{x}_j,\widehat{\xi}_k \right] = i \delta_{j,k}$, we consider the position and momentum measures $|f(x_1,x_2)|^2$, $|\widetilde{f} (\xi_1,\xi_2)|^2$, and the associated entropies $E\left(|f|^2\right)$ and $E \left(|\widetilde{f}|^2 \right)$. We go from the position to the momentum representation via the Fourier transform $\widetilde{f} = \mathcal{F} f$. On the other hand, for our noncommutative non-canonical algebra, $\widehat{q}_1$ does not commute with $\widehat{q}_2$ nor with $\widehat{p}_1$, so it is not at all clear how one could construct entropy measures for these variables. Even if we consider the isomorphism (\ref{eq3.5}) and express them in terms of the Heisenberg-Weyl variables, the isomorphism is non-linear.
\end{enumerate}

Let then $\alpha = (\widehat{u}, \widehat{v}) \in \mathcal{N}$, and let us define the functional:
\begin{equation}
F^{(\alpha)} : \mathcal{B}^{\alpha} (\mathbb{R}^2) \to \mathbb{R}, \hspace{1 cm} F^{(\alpha)} \left[f \right] := \|\widehat{u} f \|_{L^2 (\mathbb{R}^2)}^2 + ||\widehat{v} \psi||_{L^2 (\mathbb{R}^2)}^2.
\label{eq6.2}
\end{equation}

\subsection{Existence of minima}

In this section, we start by proving that for each $\alpha \in \mathcal{N}$, there exists $f_0 \in \mathcal{B}^{\alpha} (\mathbb{R}^2) $ with $\|f_0 \|_{L^2 (\mathbb{R}^2)} =1$ minimizing $F^{(\alpha)}$, that is:
\begin{equation}
F^{(\alpha)} \left[f_0 \right] \leq F^{(\alpha)} \left[f \right] , \hspace{1 cm~} \mbox{ for all } f \in \mathcal{B}^{\alpha} (\mathbb{R}^2) \mbox{ with } \| f \|_{L^2 (\mathbb{R}^2)} =1.
\label{eq6.1.1}
\end{equation}
We thus look for the minimizer of $F^{(\alpha)}$ in the set
\begin{equation}
S:=   \left\{f \in \mathcal{B}^{\alpha} (\mathbb{R}^2): ~ \| f \|_{L^2 (\mathbb{R}^2)} =1 \right\}.
\label{eq6.1.2}
\end{equation}

\vspace{0.3 cm}
\noindent
Let us denote by $\overline{B}_R^{(\alpha)}$ the closed ball of radius $R>0$ in $\mathcal{B}^{\alpha} (\mathbb{R}^2)$:
\begin{equation}
\overline{B}_R^{(\alpha)}:= \left\{f \in \mathcal{B}^{\alpha} (\mathbb{R}^2): ~\| f\|_{\alpha} \leq R \right\}.
\label{eqClosedBall}
\end{equation}

\begin{Proposition}\label{PropositionCompactnessHilbert1}
Let $R>0$. If the set
\begin{equation}
U_R^{(\alpha)}= \left\{ f \in \overline{B}_R^{(\alpha)}: ~ \|f \|_{L^2(\mathbb{R}^2)} =1 \right\}
\label{eqCompactnessHilbert2}
\end{equation}
is nonempty, then it is a weakly sequentially compact subset of $\mathcal{B}^{\alpha} (\mathbb{R}^2)$.
\end{Proposition}

\begin{proof}
Suppose that $U_R^{(\alpha)}$ is nonempty. Let $(f_n)_n$ be an arbitrary sequence in $U_R^{(\alpha)}$. Since $\mathcal{B}^{\alpha} (\mathbb{R}^2)$ is reflexive (it is a Hilbert space), we conclude that $(f_n)_n$ has a weakly convergent subsequence $(g_k)_k$, say
\begin{equation}
g_k \rightharpoonup g,
\label{eqCompactnessHilbert3}
\end{equation}
for some $g \in \mathcal{B}^{\alpha} (\mathbb{R}^2)$, and by Mazur's Theorem $g \in  \overline{B}_R^{(\alpha)}$. It remains to prove that $\|g \|_{L^2 (\mathbb{R}^2)}=1$. From Theorem \ref{TheoremCompactEmbeddingD} (see Appendix B), we conclude that the sequence $(g_k)_k$ has a subsequence $(h_l)_l$ converging strongly in $ L^2 (\mathbb{R}^2)$, say $\|h_l -h \|_{L^2(\mathbb{R}^2)} \to 0$, for some $h \in L^2(\mathbb{R}^2)$. By the continuity of the norm, we also have $\|h \|_{L^2(\mathbb{R}^2)} =1$. The proof is complete if we show that $g=h$ a.e.

The mapping $(u,v) \mapsto \langle u,v\rangle_{L^2 (\mathbb{R}^2)}$ is a sesquilinear form on $\mathcal{B}^{\alpha} (\mathbb{R}^2) \times \mathcal{B}^{\alpha} (\mathbb{R}^2)$. Moreover, it is bounded as we now prove. From the orthogonality relations (\ref{eqModulationspace4}) and the Cauchy-Schwarz inequality, we have:
\begin{equation}
\begin{array}{c}
\left|\langle u,v\rangle_{L^2 (\mathbb{R}^2)} \right| =\frac{1}{\|g\|_{L^2 (\mathbb{R}^2)}^2} \left| \langle V_g u, V_g v \rangle_{L^2 (\mathbb{R}^4)}\right| \\
\\
\precsim \| V_g u\|_{L^2 (\mathbb{R}^4)} ~ \| V_g v\|_{L^2 (\mathbb{R}^4)} \\
\\
\precsim \|m V_g u\|_{L^2 (\mathbb{R}^4)} ~ \|m V_g v\|_{L^2 (\mathbb{R}^4)} = \|u\|_{M_m^2 (\mathbb{R}^2)} ~ \|v\|_{M_m^2 (\mathbb{R}^2)} \asymp \|u\|_{\alpha} ~ \|v\|_{\alpha},
\end{array}
\label{eqCompactnessHilbert5}
\end{equation}
where $m$ is given by (\ref{eqDModulation2}). By the Riesz representation theorem, there exists a bounded linear operator $\widehat{A}:\mathcal{B}^{\alpha} (\mathbb{R}^2) \to \mathcal{B}^{\alpha} (\mathbb{R}^2)$, such that
\begin{equation}
\langle u,v\rangle_{L^2 (\mathbb{R}^2)}=\langle\widehat{A}  u,v\rangle_{\alpha},
\label{eqCompactnessHilbert6}
\end{equation}
for all $u,v \in \mathcal{B}^{\alpha} (\mathbb{R}^2)$.

Let $u \in \mathcal{S} (\mathbb{R}^2)$. We then have:
\begin{equation}
\langle h_l-g,u\rangle_{L^2 (\mathbb{R}^2)}=\langle\widehat{A} ( h_l-g), u\rangle_{\alpha}.
\label{eqCompactnessHilbert7}
\end{equation}
If we take the limit $l \to \infty$, the right-hand side of the previous equation vanishes, while the left-hand side becomes $\langle h-g,u\rangle_{L^2 (\mathbb{R}^2)}$. Since $\mathcal{S} (\mathbb{R}^2)$ is dense in $L^2 (\mathbb{R}^2)$, we conclude that $h=g$ a.e.
\end{proof}

\begin{Lemma}\label{Lemmaweaklowersemicontinuous}
Let $R>0$ be such that $U_R^{(\alpha)}$ as defined in Proposition \ref{PropositionCompactnessHilbert1} is nonempty. Then, the functional $F^{(\alpha)}$, given by (\ref{eq6.2}), is weakly lower semicontinuous in $U_R^{(\alpha)}$.
\end{Lemma}

\begin{proof}
Consider the maps
\begin{equation}
f(x) \mapsto \left(\widehat{u}f \right)(x), \hspace{0.5 cm} f (x) \mapsto  \left(\widehat{v}f \right)(x),
\label{eqOptimizers1}
\end{equation}
for $x \in \mathbb{R}^2$ and $f \in \mathcal{S} (\mathbb{R}^2)$. They extend to bounded linear operators in $\mathcal{B}^{\alpha} (\mathbb{R}^{2}) = \mathcal{B} (\mathbb{R}^{2}) = M_m^2 (\mathbb{R}^2)$. Thus, the map $f \mapsto \|\widehat{u} f\|_{L^2(\mathbb{R}^2)}$ is a continuous and convex functional on the closed ball $\overline{B}_R^{(\alpha)}$, which is a convex and closed subset of $\mathcal{B}^{(\alpha)} (\mathbb{R}^{2})$. This means that this map is weakly lower semicontinuous in $\overline{B}_R^{(\alpha)}$.

On the other hand, by Proposition \ref{PropositionCompactnessHilbert1}, $U_R^{(\alpha)}$ is a weakly sequentially compact subset of $\mathcal{B}^{\alpha} (\mathbb{R}^2)$. Hence the restriction of $f \mapsto \|\widehat{u} f\|_{L^2(\mathbb{R}^2)}$ to $U_R^{(\alpha)} \subset \overline{B}_R^{(\alpha)}$ is  weakly lower semicontinuous. The product of two nonnegative weakly lower semicontinuous functionals is again weakly lower semicontinuous, which entails that $f \mapsto \|\widehat{u}  f\|_{L^2(\mathbb{R}^2)}^2$ is weakly lower semicontinuous. The same can be said about $f \mapsto  \|\widehat{v}  f\|_{L^2(\mathbb{R}^2)}^2$. Consequently, $F^{(\alpha)}$, being the sum of these two functionals, is weakly lower semicontinuous in $U_R^{(\alpha)}$.
\end{proof}

\vspace{0.3 cm}
\noindent
We next prove the existence of minimizers.
\begin{Theorem}\label{TheoremMinimizersHilbert}
Let $R>1$ be such that $U_R^{(\alpha)}$ as defined in Proposition \ref{PropositionCompactnessHilbert1} is nonempty. Then there exists $f_0 \in U_R^{(\alpha)}$ such that
\begin{equation}
F^{(\alpha)} \left[f_0 \right]  \leq F^{(\alpha)} \left[f \right],
\label{eqMinimizersHilbert2}
\end{equation}
for all
\begin{equation}
f \in S:= \left\{f \in \mathcal{B}^{\alpha} (\mathbb{R}^2): ~ \|f\|_{L^2 (\mathbb{R}^2)} =1 \right\}.
\label{eqMinimizersHilbert3}
\end{equation}
\end{Theorem}

\begin{proof}
The set $U_R^{(\alpha)}$ is weakly sequentially compact (cf. Proposition \ref{PropositionCompactnessHilbert1}). Moreover, the functional $F^{(\alpha)} $ is weakly lower semicontinuous. Consequently, there exists a minimizer $f_0$ of $F^{(\alpha)} $ in $U_R^{(\alpha)}$. It remains to prove that $f_0$ is in fact a minimizer on the whole set $S$.

From (\ref{eq6.2}) and (\ref{eqMinimizersHilbert3}), we have:
\begin{equation}
F^{(\alpha)} \left[f \right]= \| f \|_{\alpha}^2 -2 \|f \|_{L^2 (\mathbb{R}^2)}^2 =\| f \|_{\alpha}^2 -2 ,
\label{eqMinimizersHilbert4}
\end{equation}
for all $f \in S$.

Since $f_0 \in U_R^{(\alpha)}$ , it follows:
\begin{equation}
F^{(\alpha)} \left[f_0 \right] \leq R^2- 2.
\label{eqMinimizersHilbert5}
\end{equation}
On the other hand, if $f \in S \backslash U_R^{(\alpha)}$:
\begin{equation}
F^{(\alpha)} \left[f \right] > R^2- 2,
\label{eqMinimizersHilbert6}
\end{equation}
and the result follows.
\end{proof}

\subsection{Euler-Lagrange equations}

We want to minimize the functional $F^{(\alpha)} \left[f \right] = \| f \|_{\alpha}^2 -2 \|f \|_{L^2 (\mathbb{R}^2)}^2$ in $\mathcal{B}^{\alpha} (\mathbb{R}^2)$, subject to the constraint:
\begin{equation}
\|f \|_{L^2 (\mathbb{R}^2)} = 1.
\label{eqEulerLagrange1}
\end{equation}
We thus optimize the functional
\begin{equation}
\mathfrak{L}^{(\alpha)} \left[f , \gamma \right] = F^{(\alpha)} \left[f  \right] + \gamma \left( 1- \|f \|_{L^2 (\mathbb{R}^2)}^2  \right),
\label{eqEulerLagrange2}
\end{equation}
where $\gamma$ is a Lagrange multiplier.

Before we proceed, let us recall that the operator $\widehat{A}$ defined in eq. (\ref{eqCompactnessHilbert6}) is such that:
\begin{equation}
\langle f, g \rangle_{L^2 (\mathbb{R}^2)}= \langle \widehat{A}f,g \rangle_{\alpha},
\label{eqEulerLagrange2.A}
\end{equation}
for all $f,g \in \mathcal{B}^{\alpha} (\mathbb{R}^2)$.

\begin{Theorem}\label{TheoremSpectrum}
The operator $\widehat{A}$ is positive-definite, compact and closed. It has empty residual spectrum, $0$ belongs to the continuous spectrum and it is a point of accumulation. Moreover, all remaining spectral values are eigenvalues.
\end{Theorem}

\begin{proof}
From the definition of $\widehat{A}$, we have for all $f,g \in \mathcal{B}^{\alpha} (\mathbb{R}^2)$:
\begin{equation}
\langle \widehat{A}f, g \rangle_{\alpha} = \langle f, g \rangle_{L^2 (\mathbb{R}^2)} = \overline{\langle g,f \rangle}_{L^2 (\mathbb{R}^2)} = \overline{\langle \widehat{A}g,f \rangle }_{\alpha} = \langle f,\widehat{A}g \rangle_{\alpha} .
\label{eqSpectrum1}
\end{equation}
Hence $\widehat{A}=\widehat{A}^{\ast}$.

Similarly
\begin{equation}
\langle \widehat{A} f,f \rangle_{\alpha} = \|f \|_{L^2 (\mathbb{R}^2)}^2 >0,
\label{eqSpectrum2}
\end{equation}
for all $f \in \mathcal{B}^{\alpha} (\mathbb{R}^2) \backslash \left\{0 \right\}$.
Consequently, $\widehat{A}$ is positive definite.

That $\widehat{A}$ is closed is a simple consequence of the fact that it is bounded and defined on the whole of $\mathcal{B}^{\alpha} (\mathbb{R}^2)$.

Next, we prove compactness. Let $\left(f_n \right)_{n \in \mathbb{N}}$ be a
bounded sequence in $\mathcal{B}^{\alpha} (\mathbb{R}^2)$:
\begin{equation}
\|f_n\|_{\alpha} \leq C
\label{eqSpectrum3}
\end{equation}
for some constant $C>0$ and all $n \in \mathbb{N}$. Since $\mathcal{B}^{\alpha} (\mathbb{R}^2) \subset
\subset L^2 (\mathbb{R}^2)$, $\left(f_n \right)_{n \in \mathbb{N}}$ has a subsequence $\left(g_n \right)_{n \in \mathbb{N}}$ which converges in $L^2 (\mathbb{R}^2)$.
It follows that
\begin{equation}
\begin{array}{c}
\|\widehat{A}g_n-\widehat{A}g_m\|_{\alpha}^2 = \langle g_n-g_m,\widehat{A}(g_n-g_m) \rangle_{L^2(\mathbb{R}^2)} \leq\\
\\
\leq \| \widehat{A} (g_n-g_m)\|_{L^2(\mathbb{R}^2)} \| g_n-g_m \|_{L^2(\mathbb{R}^2)}  \leq \\
\\
\leq \| \widehat{A} (g_n-g_m)\|_{\alpha}
 \|g_n-g_m \|_{L^2(\mathbb{R}^2)} \leq\\
 \\
 \leq 2 C \|\widehat{A} \|_{Op} ~ \|g_n-g_m \|_{L^2(\mathbb{R}^2)}
\end{array}
\label{eqSpectrum4}
\end{equation}
This shows that $\left(\widehat{A}g_n \right)_{n \in \mathbb{N}}$ is a Cauchy sequence.
 Since $\mathcal{B}^{\alpha} (\mathbb{R}^2)$ is complete, we conclude that
 $\left(\widehat{A}g_n \right)_{n \in \mathbb{N}}$ converges. Since the bounded sequence
 $\left(f_n \right)_{n \in \mathbb{N}}$ was chosen arbitrarily, the operator $\widehat{A}$
 is compact.

That all non-zero elements of the spectrum are eigenvalues is an immediate consequence of the fact that $\widehat{A}$ is compact. Eq. (\ref{eqSpectrum2}) shows that $\widehat{A}$ is injective. Since $\widehat{A}$ is compact and injective, we conclude that $0$ is in the continuous spectrum and the residual spectrum is empty.

If the spectrum of $\widehat{A}$ were finite, then $\widehat{A}$ would have to be of finite rank. But since $\widehat{A}: \mathcal{B}^{\alpha} (\mathbb{R}^2) \to \text{Ran}(\widehat{A})$ is bijective, this is impossible. Hence, the spectrum is infinite and $0$ must be an accumulation point.
\end{proof}

The operator $\widehat{A}$ is invertible. Its inverse has the following properties.

\begin{Theorem}\label{TheoremInverse}
$\widehat{A}^{-1}$ is densely defined in $\mathcal{B}^{\alpha} (\mathbb{R}^2)$, closed and positive-definite. Its spectrum consists only of eigenvalues which can be written as a sequence $0 < \nu_1 \leq \nu_2 \leq \cdots$, with $\nu_j \to + \infty$. Moreover, all the eigenspaces are finite dimensional.
\end{Theorem}

\begin{proof}
We start by proving that $\text{Ran}(\widehat{A})$ is dense in $ \mathcal{B}^{\alpha} (\mathbb{R}^2)$. Since $\mathcal{B}^{\alpha} (\mathbb{R}^2)=M_m^2 (\mathbb{R}^2)$, with $m \in \mathcal{P} (\mathbb{R}^{4})$, then $\mathcal{S}(\mathbb{R}^{2}) \subset \mathcal{B}^{\alpha} (\mathbb{R}^2) $. Let $f \in \mathcal{B}^{\alpha} (\mathbb{R}^2)$ be such that:
\begin{equation}
0= \langle \widehat{A}g,f \rangle_{\alpha} = \langle g,f \rangle_{L^2 (\mathbb{R}^2)},
\label{eqTheoremInverse1}
\end{equation}
for all $g \in \mathcal{S}(\mathbb{R}^{2})$. Since $\mathcal{S}(\mathbb{R}^{2})$ is dense in $L^2 (\mathbb{R}^2)$, we conclude that $f=0$, and thus $\left\{\widehat{A}g: ~g \in \mathcal{S}(\mathbb{R}^{2}) \right\}$ is dense in $ \mathcal{B}^{\alpha} (\mathbb{R}^2)$.

Under these circumstances and taking into account Theorem \ref{TheoremSpectrum}, we conclude that $\widehat{A}^{-1}$ is densely defined, closed and that $(\widehat{A}^{-1})^{\ast}=(\widehat{A}^{\ast})^{-1}$. Thus $\widehat{A}^{-1}$ is self-adjoint and positivity follows immediately.

The statements regarding the spectrum are also an immediate consequence of Theorem \ref{TheoremSpectrum}. We just remark that $0$ is a regular value of $\widehat{A}^{-1}$, since its inverse $\widehat{A}$ exists, is bounded and defined on the whole $ \mathcal{B}^{\alpha} (\mathbb{R}^2)$.
\end{proof}

We are now in a position to obtain the Euler-Lagrange equation for the minimizer $f_0$.

\begin{Theorem}\label{theoremEulerLagrangeEquations}
Let $f_0 \in S$ satisfy
\begin{equation}
 F^{(\alpha)} [f_0] = \text{min}_{f \in S} F^{(\alpha)} \left[f \right].
\label{eq6.2.1}
\end{equation}
Then
\begin{equation}
\langle \widehat{u}g, \widehat{u}f_0 \rangle_{L^2 (\mathbb{R}^2)} + \langle \widehat{v}g, \widehat{v}f_0 \rangle_{L^2 (\mathbb{R}^2)} = F^{(\alpha)}[f_0] \langle g,f_0 \rangle_{L^2 (\mathbb{R}^2)},
\label{eq6.2.2}
\end{equation}
for all $g \in \mathcal{B}^{\alpha} (\mathbb{R}^2)$.
\end{Theorem}

\begin{proof}
Using standard techniques in variational calculus \cite{Evans,Jahn,Jost} we obtain, from the Fr\'echet derivative of the functional (\ref{eqEulerLagrange2}), the following stationarity condition :
\begin{equation}
\begin{array}{c}
\langle \widehat{u}g, \widehat{u}f_0 \rangle_{L^2 (\mathbb{R}^2)} + \langle \widehat{u}f_0, \widehat{u}g \rangle_{L^2 (\mathbb{R}^2)} + \langle \widehat{v}g, \widehat{v}f_0 \rangle_{L^2 (\mathbb{R}^2)} + \langle \widehat{v}f_0, \widehat{v}g \rangle_{L^2 (\mathbb{R}^2)}= \\
\\
=\gamma \langle g,f_0 \rangle_{L^2 (\mathbb{R}^2)} + \gamma \langle f_0,g \rangle_{L^2 (\mathbb{R}^2)},
\end{array}
\label{eq6.2.3}
\end{equation}
for all $g \in \mathcal{B}^{\alpha} (\mathbb{R}^2)$.

Since the previous equation holds for all $g \in \mathcal{B}^{\alpha} (\mathbb{R}^2)$, then in particular it holds for $ig$. Adding eq.(\ref{eq6.2.3}) for $ig$ and the same equation multiplied by $i$, we obtain:
\begin{equation}
\langle \widehat{u}g, \widehat{u}f_0 \rangle_{L^2 (\mathbb{R}^2)} + \langle \widehat{v}g, \widehat{v}f_0 \rangle_{L^2 (\mathbb{R}^2)} = \gamma \langle g,f_0 \rangle_{L^2 (\mathbb{R}^2)},
\label{eq6.2.4}
\end{equation}
for all $g \in \mathcal{B}^{\alpha} (\mathbb{R}^2)$.

Upon substitution of $g= f_0$ in the previous equation, we obtain $\gamma = F^{(\alpha)}[f_0]$.
\end{proof}

\begin{Corollary}\label{cor11}
Let $f_0$ be a minimizer of $F^{(\alpha)}$ on $S$. Then $f_0$ is an eigenvector of the operator $\widehat{H}^{(\alpha)} = \widehat{u}^2 + \widehat{v}^2$ associated with the smallest eigenvalue $\nu_0= F^{(\alpha)} [f_0] $. In other words $f_0$ is the fundamental or ground state of the positive operator $\widehat{H}^{(\alpha)}$. Moreover, the operator $\widehat{H}^{(\alpha)}$ can be identified with $\widehat{A}-2 \widehat{I}$ in $ \text{Dom} (\widehat{H}^{(\alpha)}) $.
\end{Corollary}

\begin{proof}
From eq.(\ref{eq6.2.4}) it is clear that $f_0$ is a (weak) solution of the eigenvalue equation:
\begin{equation}
\left(\widehat{u}^2 + \widehat{v}^2 \right) f_0= \gamma f_0.
\label{eqEigenvalueequation1}
\end{equation}

Now let $f_{\nu}$ be any other eigenvector of $\widehat{H}^{(\alpha)} = \widehat{u}^2 + \widehat{v}^2$ with eigenvalue $\nu$:
\begin{equation}
\widehat{H}^{(\alpha)} f_{\nu} = \nu f_{\nu}
\label{eq6.20}
\end{equation}
If we perform the inner product with $f_{\nu}$ in the previous equation, we obtain:
\begin{equation}
\|\widehat{u} f_{\nu} \|_{L^2 (\mathbb{R}^2)}^2 + \|\widehat{v} f_{\nu} \|_{L^2 (\mathbb{R}^2)}^2 = \nu \| f_{\nu} \|_{L^2 (\mathbb{R}^2)}^2
\label{eq6.21}
\end{equation}
If $f$ is normalized, we obtain:
\begin{equation}
\nu = F^{(\alpha)} \left[ f_{\nu} \right]
\label{eq6.22}
\end{equation}
which shows that the eigenvalue of an eigenvector is equal to the value of $F^{(\alpha)}$ for that eigenvector. It then follows that:
\begin{equation}
\nu =  F^{(\alpha)} \left[ f_{\nu} \right] \geq  F^{(\alpha)} \left[f_0 \right]
\label{eq6.23}
\end{equation}
where we used the fact that $f_0$ is the minimizer of $F^{(\alpha)} $. We conclude that $ F^{(\alpha)} \left[ f_0 \right]$ is the smallest eigenvalue of $\widehat{H}^{(\alpha)}$.

Finally, let us establish the connection with operator $\widehat{A}$. We can rewrite the Euler-Lagrange equations (\ref{eq6.2.4}) as:
\begin{equation}
\langle g,f_0 \rangle_{\alpha} = (\gamma+2) \langle g, f_0 \rangle_{L^2 (\mathbb{R}^2)}.
\label{eqEulerLag1}
\end{equation}
But, in view of the definition of the operator $\widehat{A}$ (\ref{eqCompactnessHilbert6}), we can rewrite this as:
\begin{equation}
\langle g,f_0 \rangle_{\alpha} = (\gamma +2) \langle \widehat{A} g, f_0 \rangle_{ \alpha} \Leftrightarrow \langle g,f_0 \rangle_{\alpha} = (\gamma+2) \langle  g, \widehat{A} f_0 \rangle_{ \alpha},
\label{eqEulerLag2}
\end{equation}
for all $g \in \mathcal{B}^{\alpha} (\mathbb{R}^2)$. And thus:
\begin{equation}
f_0 =(\gamma +2) \widehat{A} f_0 \Leftrightarrow \left(\widehat{A}^{-1} - 2 \widehat{I} \right)f_0= \gamma f_0.
\label{eqEulerLag2}
\end{equation}
This shows that $f_0$ is also an eigenvector of $\widehat{A}^{-1} - 2 \widehat{I}$ with the same eigenvalue as $\widehat{H}^{(\alpha)}$. In fact the two operators are the same in $\text{Dom} (\widehat{H}^{(\alpha)})$. Indeed, let $h \in  \text{Dom} (\widehat{H}^{(\alpha)})$. Then for all $g \in \mathcal{B}^{\alpha} (\mathbb{R}^2)$, we have:
\begin{equation}
\begin{array}{c}
\langle \widehat{A} \left(\widehat{H}^{(\alpha)} + 2 \widehat{I} \right) h , g \rangle_{\alpha} = \langle \left(\widehat{H}^{(\alpha)} + 2 \widehat{I} \right) h , g \rangle_{L^2 (\mathbb{R}^2)} =\\
\\
=\langle \left(\widehat{u}^2 + \widehat{v}^2 +2 \widehat{I}\right) h , g \rangle_{L^2 (\mathbb{R}^2)} = \\
\\
= \langle \widehat{u}  h , \widehat{u}  g \rangle_{L^2 (\mathbb{R}^2)}+\langle \widehat{v}  h , \widehat{v}  g \rangle_{L^2 (\mathbb{R}^2)} +2 \langle h ,   g \rangle_{L^2 (\mathbb{R}^2)} = \langle h , g \rangle_{\alpha}.
\end{array}
\label{eqEulerLag3}
\end{equation}
And thus: $\widehat{A} \left(\widehat{H}^{(\alpha)} + 2 \widehat{I} \right) h =h$, for all $h \in \text{Dom} (\widehat{H}^{(\alpha)})$, which proves the result.
\end{proof}

\subsection{Saturation of the inequalities}

As we mentioned in the introduction, it would be interesting to determine whether there are states minimizing all uncertainty relations for this non-canonical noncommutative algebra. In other words, is there a state $f_0 \in \mathcal{B} (\mathbb{R}^2)$ such that $f_0$ is an eigenvector associated with the smallest eigenvalue of $\widehat{H}^{(\alpha)}$ for all $\alpha=(\widehat{u}, \widehat{v})\in \mathcal{N}$? Or, equivalently, is there a state $f_0$ such that:
\begin{equation}
\|\widehat{u} f_0 \|_{L^2 (\mathbb{R}^2)}^2 + \|\widehat{v} f_0 \|_{L^2 (\mathbb{R}^2)}^2 \leq \|\widehat{u} f \|_{L^2 (\mathbb{R}^2)}^2 + \|\widehat{v} f \|_{L^2 (\mathbb{R}^2)}^2,
\label{eqCoherent1}
\end{equation}
for all $\alpha=(\widehat{u}, \widehat{v})\in \mathcal{N}$ and all $f \in  \mathcal{B} (\mathbb{R}^2)$?

This is an open problem, but we conjecture that this is impossible. The {\it rationale} for this conjecture is that there are several hindrances to such a coherent state:

\vspace{0.2 cm}
\begin{itemize}
\item Given $\alpha=(\widehat{u}_1, \widehat{v}_1),~\beta=(\widehat{u}_2, \widehat{v}_2)\in \mathcal{N}$, with $\alpha \neq \beta$, it can be shown that $\widehat{H}^{(\alpha)}$ and $\widehat{H}^{(\beta)}$ do not commute. Therefore there is no common orthonormal set of eigenvectors.

\item The previous statement does not preclude the existence of a common eigenvector $f_0$ of $\widehat{H}^{(\alpha)}$ and $\widehat{H}^{(\beta)}$, provided $f_0 \in \text{Ker} \left( \left[\widehat{H}^{(\alpha)} ,\widehat{H}^{(\beta)} \right]\right)$. However, $f_0$ would have to be a common eigenvector of $\widehat{H}^{(\alpha)} $ for all $\alpha \in \mathcal{N}$, and the associated eigenvalue would have to be the smallest eigenvalue of $\widehat{H}^{(\alpha)} $ for all $\alpha \in \mathcal{N}$.

\item If we take the limit $\epsilon \to 0^+$, we recover a canonical phase space noncommutative algebra. It is known that, in this case, there is no state which saturates two uncertainty relations simultaneously \cite{Bolonek,Kosinski}.
\end{itemize}

\vspace{0.3 cm}
\noindent
Although we have in principle a way of determining the minimizers of the various uncertainty relations, in practise we cannot obtain analytic solutions of the partial differential equations. For example, in the case $u=q_1$, $v=q_2$, $H_{q_1,q_2}$ has the differential representation:
\begin{equation}
- \left(\frac{\theta}{2 \lambda} \right)^2 \Delta - i \theta x_2 \frac{\partial }{\partial x_1} +i \theta \left(x_1 + \frac{E}{\lambda}x_1^2 \right) \frac{\partial }{\partial x_2} + \lambda^2 |x|^2 + 2 \lambda E x_1^3 + E^2 x_1^4
\label{eq6.24}
\end{equation}
where $|x|^2 = x_1^2 + x_2^2$ and $\Delta = \frac{\partial^2 }{\partial x_1^2} + \frac{\partial^2 }{\partial x_2^2}$ is the Laplacian. This illustrates the difficulty in obtaining the minimizers and checking the existence of coherent states.

\subsection{The Heisenberg-Pauli-Weyl inequality}

One of the interesting aspects of our algebra is the fact that the usual Heisenberg-Pauli-Weyl inequality of position and momentum can be violated.

\begin{Example}
Consider the Gaussian state:
\begin{equation}
f (x_1, x_2) = \left(\frac{4}{\pi^2 ab} \right)^{\frac{1}{4}} \exp \left[ - \frac{1}{a} (x_1 -x_1^{(0)})^2 - \frac{1}{b} (x_2 - x_2^{(0)})^2 \right]
\label{eq6.25}
\end{equation}
where $a,b >0$ and
\begin{equation}
x_1^{(0)} = - \frac{\lambda}{2E},
\label{eq6.24}
\end{equation}
for $E \ne 0 \Rightarrow \theta,\epsilon \ne 0$. The expectation value of $\widehat{q}_1$ in this state is
\begin{equation}
\begin{array}{c}
\langle \widehat{q}_1f,  f \rangle_{L^2 (\mathbb{R}^2)}= \int_{\mathbb{R}^2}  \left[\left(\lambda x_1 + \frac{i \theta}{2 \lambda} \frac{\partial}{\partial x_2}  + E x_1^2 \right)f (x_1, x_2) \right] \overline{f (x_1, x_2)} dx_1  dx_2 =\\
\\
= \lambda x_1^{(0)} + E \left((x_1^{(0)})^2 + \frac{a}{4} \right)
\end{array}
\label{eq6.25}
\end{equation}
and that of $\widehat{q}_1^2$ reads
\begin{equation}
\begin{array}{c}
\langle \widehat{q}_1^2 f,f \rangle_{L^2 (\mathbb{R}^2)}=E^2 (x_1^{(0)})^4 + 2 \lambda E (x_1^{(0)})^3 + \left(\lambda^2 + \frac{3}{2} a E^2 \right) (x_1^{(0)})^2 + \frac{3}{2} \lambda a E x_1^{(0)} + \\
\\
+\frac{\lambda^2 a}{4} + \frac{\theta^2}{4 \lambda^2 b} + \frac{3 a^2 E^2}{16}
\end{array}
\label{eq6.26}
\end{equation}
One thus obtains the dispersion
\begin{equation}
\Delta_{q_1} (f, <q_1>_f) = \left[ a \left( E x_1^{(0)} + \frac{\lambda}{2} \right)^2 + \frac{\theta^2}{4 \lambda^2 b} + \frac{a^2 E^2}{8} \right]^{\frac{1}{2}}
\label{6.27}
\end{equation}
It follows from (\ref{eq6.24}) that
\begin{equation}
\Delta_{q_1} (f, <q_1>_f) =\left( \frac{\theta^2}{4 \lambda^2 b} + \frac{a^2 E^2}{8} \right)^{\frac{1}{2}}
\label{eq6.28}
\end{equation}
In a similar fashion one obtains:
\begin{equation}
\langle \widehat{p}_1 f ,f \rangle_{L^2 (\mathbb{R}^2)}= \int_{\mathbb{R}^2} \left[\left(-i \mu \frac{\partial}{\partial x_1}  + \frac{\eta}{2 \mu} x_2 \right)f (x_1, x_2) \right] \overline{f (x_1, x_2)}  dx_1  dx_2 = \frac{\eta}{2 \mu} x_2^{(0)}
\label{eq6.29}
\end{equation}
and
\begin{equation}
\langle \widehat{p}_1^2 f ,f \rangle_{L^2 (\mathbb{R}^2)}= \frac{\mu^2}{a} + \left(\frac{\eta}{2 \mu} \right)^2 \left((x_2^{(0)})^2 + \frac{b}{4} \right)
\label{eq6.30}
\end{equation}
It follows that
\begin{equation}
\Delta_{p_1} (f, <p_1>_f) = \frac{\mu}{\sqrt{a}} \left[1 +  \left(\frac{\eta}{4 \mu^2} \right)^2 ab \right]^{\frac{1}{2}}
\label{eq6.31}
\end{equation}
From (\ref{eq6.28},\ref{eq6.31})
\begin{equation}
\Delta_{q_1} (f, <q_1>_f) \Delta_{p_1} (f, <p_1>_f) = \frac{1}{2} \left( \frac{1}{2} \mu^2 E^2 a + \frac{\mu^2 \theta^2}{\lambda^2 ab} + \frac{\eta^2 E^2 a^2 b}{32 \mu^2} + \frac{\theta^2 \eta^2}{16 \mu^2 \lambda^2} \right)^{\frac{1}{2}}
\label{eq6.32}
\end{equation}
If one chooses $b = a^{- 3/2}$ and lets $a \downarrow 0$, then:
\begin{equation}
\Delta_{q_1} (f, <q_1>_f) \Delta_{p_1} (f, <p_1>_f) \to \frac{\theta \eta}{8 \mu \lambda} = \frac{\xi}{4 (1 + \sqrt{1- \xi})} \leq \frac{\xi}{4} < \frac{1}{2}
\label{eq6.33}
\end{equation}
And thus (\ref{eq6.33}) violates the usual position-momentum uncertainty relations. This is an interesting feature of this new algebra, which makes a clear distinction with the standard Heisenberg-Weyl algebra. This should not pose interpretational problems, because (as we mentioned before) $(q_1,q_2)$ and $(p_1,p_2)$ may represent other physical quantities which are not the usual position and momentum of particles (e.g. the scale factors in the Kantowski-Sachs model and their conjugate momenta).
\end{Example}

\subsection{On minimal length and momentum}

Recently, there has been a great deal of work devoted to quantum theories with minimal length \cite{Dey1,Dey2,Dey3,Dey4,Fityo,Kempf,Kober}. In such theories a coordinate $x$ with minimal length $L$ satisfies the condition $\Delta_x (f) \ge L$ for all normalized states $\psi$. The previous example shows that, for our algebra, there are no minimal length for $q_1$ and minimal momentum for $p_1$. Indeed, if one lets $a \downarrow 0$ and $b \to \infty$ in (\ref{eq6.28}), then $\Delta_{q_1} (f)$ can be made arbitrarily small. Alternatively, if one sets $ab=1$ and lets $a \to \infty$ in (\ref{eq6.31}), then $ \Delta_{p_1} (f)$ can also become arbitrarily small. Similar conclusions can be drawn for $q_2$ and $p_2$.

\appendix

\section{Modulation spaces}

In this appendix, we shall prove that the space $\mathcal{B} (\mathbb{R}^2)$ (or, equivalently, the spaces $\mathcal{B}^{\alpha} (\mathbb{R}^2)$) is a particular instance of a family of functional spaces called {\it modulation spaces} which find many applications in time-frequency analysis \cite{Feichtinger,Grochenig}. In the sequel $x$ denotes a time variable and $\omega$ a frequency variable. This is the more familiar interpretation in the context of modulation spaces, but we can easily switch to position and momentum.

\begin{Definition}\label{DefinitionWeight0}
A weight in $\mathbb{R}^d$ is a positive function $m \in L_{loc}^{\infty} (\mathbb{R}^d)$. Given two weights $m$ and $v$, $m$ is said to be $v$-moderate, if
\begin{equation}
m(x+y) \le C m(x) v(y), \hspace{1 cm} \forall x,y \in \mathbb{R}^d,
\label{eqvmoderate}
\end{equation}
for some $C>0$. We demote by $\mathcal{P} (\mathbb{R}^d)$ the set of all weights $m$ which are $v$-moderate for some polynomial weight $v$.
\end{Definition}

\begin{Definition}\label{DefinitionModulationspace}
Given a fixed window $g \in \mathcal{S} (\mathbb{R}^d) \backslash \left\{0 \right\}$, we define the short-time Fourier transform of $f \in  \mathcal{S} (\mathbb{R}^d)$ by
\begin{equation}
  V_gf(x,\omega)=\langle f,\pi(x,\omega)g \rangle_{L^2 (\mathbb{R}^d)}=\int_{\mathbb{R}^d} f(t)\overline{g(t-x)}e^{-2\pi it \cdot \omega}dt.
\label{eqModulationspace1}
\end{equation}
Here $\pi(x,\omega)g (t)= e^{2 \pi i t \cdot \omega} g(t- x)$ is (up to a phase) the Schr\"odinger representation of the Heisenberg group $\mathbb{H}(d)$.

This extends to $f \in  \mathcal{S}^{\prime} (\mathbb{R}^d) $, if we use the duality bracket:
\begin{equation}
V_gf(x,\omega)=\langle f,\pi(x,\omega) \overline{g} \rangle.
\label{eqModulationspace2}
\end{equation}
To study the time-frequency content of a function, we shall consider the mixed norm:
\begin{equation}
\|F\|_{L_{x,\omega}^{r,s} (\mathbb{R}^{2d})}=\left(\int_{\mathbb{R}^d} \left(\int_{\mathbb{R}^d} |F(x, \omega)|^r d x\right)^{\frac{s}{r}} d \omega\right)^{\frac{1}{s}},
\label{eqLieb4.4}
\end{equation}
for $F \in \mathcal{S}^{\prime} (\mathbb{R}^{2d})$ and $1 \le r,s< \infty$, with the obvious modification for $r$ or $s= \infty$.

Given a weight $m$, the modulation space $M_m^{r,s} (\mathbb{R}^d)$ is defined as the set of all $f \in \mathcal{S}^{\prime} (\mathbb{R}^{d})$ such that
\begin{equation}
\|f\|_{M_m^{r,s}  (\mathbb{R}^d)}= \|m V_gf\|_{L_{x,\omega}^{r,s} (\mathbb{R}^{2d})} = \left(\int_{\mathbb{R}^d} \left(\int_{\mathbb{R}^d} |V_g f(x, \omega) m (x, \omega) |^r d x\right)^{\frac{s}{r}} d \omega\right)^{\frac{1}{s}} < \infty.
\label{eqModulationspace3}
\end{equation}
We shall write $M_m^r$, when $r=s$ and $M^{r,s}$, when $m\equiv 1$.

\end{Definition}

We should remark that different choices of windows $g \in \mathcal{S} (\mathbb{R}^d) \backslash \left\{0 \right\}$ lead to equivalent norms, and that modulation spaces are Banach spaces. Moreover, if $p=q=2$, then $M_m^{2,2} (\mathbb{R}^d)$ are in fact Hilbert spaces.

Among the modulation spaces we find the following well-known spaces:

\begin{enumerate}
\item $M^2 (\mathbb{R}^d)=L^2 (\mathbb{R}^d)$. This is an immediate consequence of the orthogonality relations \cite{Grochenig}:
\begin{equation}
\langle V_{g_1} f_1, V_{g_2} f_2 \rangle_{L^2 (\mathbb{R}^{2d})} = \langle f_1,f_2 \rangle_{L^2 (\mathbb{R}^d)} \overline{\langle g_1,g_2 \rangle}_{L^2 (\mathbb{R}^d)},
\label{eqModulationspace4}
\end{equation}
which hold for all $f_1,f_2,g_1, g_2 \in L^2 (\mathbb{R}^d)$.

\vspace{0.2 cm}
\item weighted $L^2$-spaces: If $m(x, \omega)= m(x)= \left(1+ |x|^2 \right)^{s/2}$ with $s \in \mathbb{R}$, then
$$
M_m^2 (\mathbb{R}^d)= L_s^2 (\mathbb{R}^d) = \left\{ f \in \mathcal{S}^{\prime} (\mathbb{R}^d):~ f (x) \left(1+ |x|^2 \right)^{s/2} \in L^2 (\mathbb{R}^d) \right\}.
$$

\vspace{0.2 cm}
\item Bessel potential spaces (Sobolev-Hilbert spaces):  If $m(x, \omega)= m(\omega)= \left(1+ |\omega|^2 \right)^{s/2}$ with $s \in \mathbb{R}$, then
$$
M_m^2 (\mathbb{R}^d)= H^s (\mathbb{R}^d) = \left\{ f \in \mathcal{S}^{\prime} (\mathbb{R}^d):~ \mathcal{F}f  (\omega) \left(1+ |\omega|^2 \right)^{s/2} \in L^2 (\mathbb{R}^d) \right\}.
$$
\end{enumerate}

Cases (2) and (3) are particular instances of the following Proposition (Proposition 11.3.1 in \cite{Grochenig}):
\begin{Proposition}\label{PropositionModulationPartial}
Let $g \in \mathcal{S} (\mathbb{R}^d) \backslash \left\{0 \right\}$ be a window and $m$ a weight.
\begin{enumerate}
\item If $m(x, \omega) =m(x)$, then $M_m^2=L_m^2$.

\vspace{0.2 cm}
\item If $m(x, \omega) =m(\omega)$, then $M_m^2= \mathcal{F} L_m^2$.
\end{enumerate}
\end{Proposition}

In the sequel, we shall need the following compact embedding theorem for modulation spaces, which was proved by Boggiatto and Toft \cite{Boggiatto} (see also \cite{Pfeuffer}):
\begin{Theorem}[Boggiatto-Toft]\label{TheoremBoggiatto}
Assume that $m_1,m_2 \in \mathcal{P} (\mathbb{R}^{2d})$, and that $p,q \in \left[1, \infty \right]$. Then the embedding
\begin{equation}
i: M_{m_1}^{p,q} (\mathbb{R}^d) \to  M_{m_2}^{p,q} (\mathbb{R}^d)
\label{eqBoggiatto1}
\end{equation}
is compact if and only if $m_2 / m_1 \in L_0^{\infty} (\mathbb{R}^{2d})$.
\end{Theorem}

Here $L_0^{\infty} (\mathbb{R}^{2d})$ is the set of all $f \in L^{\infty} (\mathbb{R}^{2d})$ such that
\begin{equation}
\lim_{R \to \infty} \left( \mbox{ess~sup}_{|z| \ge R} |f(z)| \right)=0 ,
\label{eqweights1}
\end{equation}
where we wrote collectively $z=(x, \omega) \in \mathbb{R}^{2d}$.

\section{Compact embedding}

We now show that the space $\mathcal{B}(\mathbb{R}^2)$ (or, equivalently, all spaces $\mathcal{B}^{\alpha} (\mathbb{R}^2)$) are in fact modulation spaces.

\begin{Proposition}\label{PropositionDModulation}
Let $\psi$ and $\phi$ be the weights:
\begin{equation}
\psi (x)= \sqrt{1+ \left(x_1+ \frac{E x_1^2}{\lambda}\right)^2 +x_2^2}, \hspace{1 cm} \phi (\omega )= \sqrt{1+4 \pi^2 | \omega|^2}.
\label{eqDModulation1}
\end{equation}
Moreover, let
\begin{equation}
m(x,\omega)= \sqrt{| \psi(x)|^2 + |\phi (\omega)|^2}.
\label{eqDModulation2}
\end{equation}
We then have:
\begin{equation}
\mathcal{B} (\mathbb{R}^2)=M_m^2 (\mathbb{R}^2).
\label{eqDModulation3}
\end{equation}
\end{Proposition}

\begin{proof}
We start by remarking that (cf.(\ref{eq0.1},\ref{eq0.2})):
\begin{equation}
\widetilde{f} (\xi)= (2 \pi)^{-1} \mathcal{F}f  \left(\frac{\xi}{2 \pi} \right).
\label{eqDModulation4}
\end{equation}
It follows that:
\begin{equation}
\begin{array}{c}
\|\sqrt{1 + |\xi|^2} \widetilde{f} \|_{L^2 (\mathbb{R}^2)}^2 = \int_{\mathbb{R}^2} (1+|\xi|^2) |\widetilde{f} (\xi)|^2 d \xi =\frac{1}{(2 \pi)^2} \int_{\mathbb{R}^2} (1+|\xi|^2) \left|\mathcal{F} f \left(\frac {\xi}{2 \pi}\right)\right|^2 d \xi =\\
\\
=\int_{\mathbb{R}^2} \left(1+ (2 \pi)^2 |\omega|^2\right) |\mathcal{F} f (\omega)|^2 d \omega =  \|\phi \mathcal{F} f \|_{L^2 (\mathbb{R}^2)}^2
\end{array}
\label{eqDModulation5}
\end{equation}
Consequently, we have from Proposition \ref{PropositionModulationPartial}:
\begin{equation}
\begin{array}{c}
\|f\|_{\mathcal{B} }^2 = \|\psi f \|_{L^2 (\mathbb{R}^2)}^2 +\|\phi \mathcal{F} f \|_{L^2 (\mathbb{R}^2)}^2 \\
\\
\asymp \| f \|_{M_{\psi}^2 (\mathbb{R}^2)}^2 +\| f \|_{M_{\phi}^2 (\mathbb{R}^2)}^2=\\
\\
= \int_{\mathbb{R}_x^2} \int_{\mathbb{R}_{\omega}^2} m^2 (x, \omega) |V_g f(x, \omega)|^2 dx d \omega = \| f \|_{M_m^2 (\mathbb{R}^2)}^2 ,
\end{array}
\label{eqDModulation6}
\end{equation}
which proves the result.
\end{proof}

From this proposition and Theorem \ref{TheoremBoggiatto}, we conclude that:

\begin{Theorem}\label{TheoremCompactEmbeddingD}
We have the following compact embedding:
\begin{equation}
\mathcal{B} (\mathbb{R}^2) \subset \subset L^2 (\mathbb{R}^2).
\label{eqDModulation7}
\end{equation}
\end{Theorem}

\begin{proof}
Since, $|\psi(x)|^2 \in \mathcal{P} (\mathbb{R}_x^2)$, and $|\phi(\omega)|^2 \in \mathcal{P} (\mathbb{R}_{\omega}^2)$, there exist positive polynomials $p(x)$ and $q(\omega)$ such that, for all $x,x^{\prime}, \omega, \omega^{\prime} \in \mathbb{R}^2$:
\begin{equation}
|\psi (x+x^{\prime})|^2 \lesssim | \psi (x)|^2 p(x^{\prime}), \hspace{1 cm}|\phi (\omega+\omega^{\prime})|^2 \lesssim  |\phi(\omega)|^2 q(\omega^{\prime})
\label{eqDModulation8}
\end{equation}
Consequently:
\begin{equation}
\begin{array}{c}
m(x+x^{\prime},\omega+\omega^{\prime}) = \sqrt{ | \psi(x+x^{\prime})|^2 + | \phi (\omega+ \omega^{\prime})|^2} \\
\\
\lesssim\sqrt{|\psi(x)|^2 p (x^{\prime}) +  |\phi(\omega)|^2 q (\omega^{\prime})} \\
\\
\lesssim\sqrt{\left( |\psi(x)|^2+ |\phi(\omega)|^2\right) \left(p(x^{\prime})+ q(\omega^{\prime})\right)}\\
\\
= m(x, \omega) \sqrt{p(x^{\prime})+q(\omega^{\prime})} \lesssim m(x, \omega) \left(1+p(x^{\prime})+q(\omega^{\prime}) \right),
\end{array}
\label{eqDModulation9}
\end{equation}
which proves that $m \in \mathcal{P} (\mathbb{R}^4)$.

On the other hand, we have:
\begin{equation}
\begin{array}{c}
m(x, \omega)= \sqrt{2+ \left(x_1+ \frac{Ex_1^2}{\lambda} \right)^2 + x_2^2 + 4 \pi^2 \omega_1^2 + 4 \pi^2 \omega_2^2} \\
\\
\geq \sqrt{2+ \left(x_1+ \frac{Ex_1^2}{\lambda} \right)^2 + x_2^2 +  \omega_1^2 +  \omega_2^2}
\end{array}
\label{eqDModulation10}
\end{equation}
If we use the hyperspherical coordinates:
\begin{equation}
\left\{
\begin{array}{l}
x_1=r \cos (\phi_1)\\
\\
x_2=r \sin (\phi_1) \cos (\phi_2)\\
\\
\omega_1=r \sin (\phi_1) \sin (\phi_2) \cos (\phi_3)\\
\\
\omega_2=r \sin (\phi_1) \sin (\phi_2) \sin (\phi_3)
\end{array}
\right.
\label{eqDModulation11}
\end{equation}
with $\phi_1,\phi_2 \in \left[0, \pi \right]$, $\phi_3 \in \left[0, 2 \pi \right]$ and $r \in \left[ \right.0, \infty \left. \right)$, we obtain:
\begin{equation}
m(x, \omega) \geq \sqrt{2+ r^2 + \frac{2E}{\lambda}r^3 \cos^3(\phi_1) + \frac{E^2}{\lambda^2}r^4 \cos^4 (\phi_1)}
\label{eqDModulation12}
\end{equation}
Let us investigate the minimum of the function:
\begin{equation}
f(\alpha)= \frac{2E}{\lambda}r^3 \alpha^3 + \frac{E^2}{\lambda^2}r^4 \alpha^4 ,
\label{eqDModulation13}
\end{equation}
with $\alpha = \cos (\phi_1) \in \left[-1, 1 \right]$. A simple inspection reveals that the minimum is either
\begin{equation}
f(-1)=  \frac{E^2}{\lambda^2}r^4 - \frac{2E}{\lambda}r^3 ,
\label{eqDModulation14}
\end{equation}
or
\begin{equation}
f\left(- \frac{3 \lambda}{2 Er} \right)= - \frac{27 \lambda^2}{16 E^2}.
\label{eqDModulation15}
\end{equation}
The second possibility is only admissible, provided $- \frac{3 \lambda}{2 Er} \in \left[-1, 1 \right]$, that is $r\geq  \frac{3 \lambda}{2 E}$.

On the other hand, if $r \geq \frac{2 \lambda}{E}$, then $ \frac{E^2}{\lambda^2}r^4 - \frac{2E}{\lambda}r^3 \geq 0$. If follows that, for $r \geq \frac{2 \lambda}{E}$, the minimum of $f(\alpha)$ is $ - \frac{27 \lambda^2}{16 E^2}$.

Thus, if $R\geq \frac{2 \lambda}{E}$, then
\begin{equation}
\text{ess sup}_{|z| \geq R} \frac{1}{m(z)} \leq \frac{1}{\sqrt{R^2 +2  - \frac{27 \lambda^2}{16 E^2}}}.
\label{eqDModulation16}
\end{equation}
Since this vanishes as $R \to \infty$ and in view of (\ref{eqDModulation10}), we conclude that $\frac{1}{m} \in L_0^{\infty } (\mathbb{R}^4)$.

From Theorem \ref{TheoremBoggiatto}, and Proposition \ref{PropositionDModulation}, we have:
\begin{equation}
\mathcal{B}(\mathbb{R}^2) =M_m^2 (\mathbb{R}^2)\subset \subset M^2 (\mathbb{R}^2) =L^2 (\mathbb{R}^2).
\label{eqDModulation17}
\end{equation}

\end{proof}

\section*{Acknowledgements}

The work of N.C. Dias and J.N. Prata is supported by the Portuguese Science Foundation (FCT) grant PTDC/MAT-CAL/4334/2014. The authors would like to thank Franz Luef for drawing their attention to references \cite{Boggiatto,Galperin1,Galperin2}.


\begin{thebibliography}{999999}

\bibitem{Bertolami1}Bertolami, O., Zarro, C.: Stability conditions for a noncommutative scalar field coupled to gravity. Phys. Lett. B 673 (2009) 83-89.

\bibitem{Bertolami2}Bertolami, O., Rosa, J.G., Arag\~ao, C., Castorina, P., Zappal\`a, D.: Noncommutative gravitational quantum well. Phys. Rev. D 72 (2005) 0250108.

\bibitem{Bertolami3}Bertolami, O., Rosa, J.G., Arag\~ao, C., Castorina, P., Zappal\`a, D.: Scaling of variables and the relation between noncommutative parameters in noncommutative quantum mechanics. Mod. Phys. Lett. A 21 (2006) 795-802.

\bibitem{Bastos1}Bastos, C., Bertolami, O.: Berry phase in the gravitational quantum well and the Seiberg-Witten map. Phys. Lett. A 372 (2008) 5556-5559.

\bibitem{Bastos2}Bastos, C., Bertolami, O., Dias, N.C., Prata, J.N.: Weyl-Wigner formulation of noncommutative quantum mechanics. J. Math. Phys. 49 (2008) 072101.

\bibitem{Bastos3}Bastos, C., Dias, N.C., Prata, J.N.: Wigner measures in noncommutative quantum mechanics. Commun. Math. Phys. 299 (2010) 709-740.

\bibitem{Bastos5}Bastos, C., Bertolami, O., Dias, N.C., Prata, J.N.: Phase-space noncommutative quantum cosmology. Phys. Rev. D 78 (2008) 023516.

\bibitem{Bastos6}Bastos, C., Bertolami, O., Dias, N.C., Prata, J.N.: Black holes and phase-space noncommutativity. Phys. Rev. D 80 (2009) 124038.

\bibitem{Bastos7}Bastos, C., Bertolami, O., Dias, N.C., Prata, J.N.: The singularity problem and phase-space non-canonical noncommutativity. Rapid commun.: Phys. Rev. D 82 (2010) 041502.

\bibitem{Bastos8}Bastos, C., Bertolami, O., Dias, N.C., Prata, J.N.: Non-canonical phase-space noncommutativity and the Kantowski-Sachs singularity for black holes. Physical Review D 84 (2011) 024005.

\bibitem{Bastos9}Bastos, C., Bernardini, A., Bertolami, O., Dias, N.C., Prata, J.N.: Entanglement due to noncommutativity in the phase-space. Physical Review D 88 (2013) 085013.

\bibitem{Beckner}Beckner, W.: Inequalities in Fourier analysis. Ann. Math. 102 (1986) 159-182.

\bibitem{Bellissard}Bellissard, J., van Elst, A., Schulz-Baldes, H.: The non-commutative geometry of the quantum Hall effect. J. Math.
Phys. 35 (1994) 5373–5451.

\bibitem{Bender}Bender, C.M., Boettcher, S., Meisinger, P.N.: PT-symmetric quantum mechanics. J. Math. Phys. 40 (1999) 2201--2229.

\bibitem{Birula}Bialynicki-Birula, I., Mycielski, J.: Uncertainty relations for information entropy in wave mechanics. Commun. Math. Phys. 44 (1975) 129–132.
\bibitem{Boggiatto}Boggiatto, P., Toft, J.: Embeddings and compactness for generalized Sobolev-Shubin spaces and modulation spaces.
Appl. Anal. 84 (2005) 269--282.

\bibitem{Bolonek}Bolonek, K., Kosi\'nski, P.: On uncertainty relations in noncommutative quantum mechanics. Phys. Lett. B 547 (2002) 51-54.

\bibitem{Busch1}Busch, P., Heinonen, T., Lahti, P.: Heisenberg's uncertainty principle. Phys. Rep. 452 (2007), 155--176.

\bibitem{Busch2}Busch, P., Lahti, P., Werner,R.F.: Quantum root-mean-square error and measurement uncertainty relations. Rev. Mod. Phys. 86 (2014) 1261--1281.


\bibitem{Carroll}Carroll, S. M., Harvey, J. A., Kostelecky, V. A., Lane, C. D., Okamoto, T.: Noncommutative field theory and Lorentz
violation. Phys. Rev. Lett. 87 (2001) 141601.

\bibitem{Connes1}Connes, A.: Noncommutative geometry. Academic Press (1994).

\bibitem{Connes2}Connes, A., Douglas, M.R., Schwarz, A.: Noncommutative geometry and matrix theory: compactification on tori. JHEP 9802 (1998) 003.

\bibitem{Cowling}Cowling, M.G., Price, J.F.: Bandwidth versus time concentration: the Heisenberg-Pauli-Weyl inequality. SIAM J. Math. Anal. 15 (1984) 151-165.

\bibitem{Daubechies}Daubechies, I.: Ten lectures on wavelets. SIAM. Phildelphia (1992).

\bibitem{Dey1}Dey, S., Fring, A., Gouba, L., Castro, P.G.: Time-dependent q-deformed coherent states for generalized uncertainty relations. Phys. Rev. D 87 (2013) 084033.

\bibitem{Dey2}Dey, S.: Q-deformed noncommutative cat states and their nonclassical properties. Phys. Rev. D 91 (2015) 044024.

\bibitem{Dey3}Dey, S., Fring, A., Hussin, V.: Nonclassicality versus entanglement in a noncommutative space. Int. J. Mod. Phys. B 31 (2017) 1650248.

\bibitem{Dey4}Dey, S., Fring, Hussin, V.: A squeezed review on coherent states and nonclassicality for non-hermitean systems with minimal length. Springer Proc. Phys. 205 (2018) 209-242.

\bibitem{Delduc}Delduc, F., Duret, Q., Gieres, F., Lefran\c{c}ois, M.: Magnetic fields in noncommutative quantum mechanics. J. Phys. Conf. Ser. 103 (2008) 012020.

\bibitem{Kochan}Demetrian, M., Kochan, D.: Quantum mechanics on noncommutative plane. Acta Phys. Slov. 52 (2002) 1.

\bibitem{Douglas}Douglas, M.R., Nekrasov, N.A.: Noncommutative field theory. Rev. Mod. Phys. 73 (2001) 977.

\bibitem{Duval}Duval, C., Horvathy, P. A.: Exotic galilean symmetry in the noncommutative plane and the Landau effect. J. Phys. A:
Math. Gen. 34 (2001) 10097.

\bibitem{Evans}Evans, L.C.: Partial differential equations. American Mathematical Society (2002).

\bibitem{Faddeev}Faddeev, L.D. Asian-Pacific Phys. News 3 (1988) 21.

\bibitem{Feichtinger}Feichtinger, H.G.: Modulation Spaces: Looking Back and Ahead. Sampling Theory in Signal and Image Processing 5 (2006) 109-140.

\bibitem{Fityo}Fityo, T.V., Vakarchuk, I.O., Tkachuk, V.M.: WKB approximation in deformed space with minimal length. J. Phys. A: Math. Gen. 39 (2006) 379-387.

\bibitem{Flato}Flato, M.: Deformation view of physical theories. Czech. J. Phys. B 32 (1982) 472-475.

\bibitem{Folland2}G.B. Folland and A. Sitaram, \textit{The uncertainty principle: A mathematical survey}, Journal of Fourier Analysis and Applications
    \textbf{3} (3) (1997) 207--238.

\bibitem{Galperin1}Galperin, Y.V.: On compactness of embeddings of Fourier-Lebesgue spaces into modulation spaces. Int. J. Anal. 2013 (2013) 681573.

\bibitem{Galperin2}Galperin, Y.V., Gr\"ochenig, K.: Uncertainty principles as embeddings of modulation spaces. J. Math. Anal. Appl. 274 (2002) 181--202.

\bibitem{Gamboa}Gamboa, J., Loewe, M., and Rojas, J. C.: Noncommutative quantum mechanics. Phys. Rev. D 64 (2001) 067901.

\bibitem{Obregon}Garc\'{\i}a-Compe\'an, H., Obreg\'on, O., and Ram\'{\i}rez, C.: Noncommutative quantum cosmology. Phys. Rev. Lett. 88 (2002) 161301.

\bibitem{Grochenig}Gr\"ochenig, K.: Foundations of time-frequency analysis. Birkh\"auser Boston (2001).

\bibitem{GrochenigStudia}Gr\"ochenig, K.: An uncertainty principle related to the Poisson summation formula. Stud. Math. 121 (1996) 87-104

\bibitem{GrochenigToft1}Gr\"ochenig, K., Toft, J.: Isomorphism properties of Toeplitz operators and pseudo-differential operators between modulation spaces. J. Anal. Math. 114 (2011) 255--283.

\bibitem{GrochenigToft2}Gr\"ochenig, K., Toft, J.: The range of localization operators and lifting theorems for modulation and Bargmann-Fock spaces.
 Trans. Amer. Math. Soc. 365 (2013) 4475--4496.

\bibitem{Heisenberg} Heisenberg, W.: \"Uber den anschaulischen Inhalt der quantentheoretischen Kinematik und Mechanik. Zeitschr. Phys. 43 (1927) 172.

\bibitem{Hirschman}Hirschman Jr., I.I.: A note on entropy. Amer. J. Math. 79 (1957) 152-156.

\bibitem{Horvathy}Horvathy, P. A.: The noncommutative Landau problem. Ann. Phys. 299 (2002) 128.

\bibitem{Jacobson}Jacobson, N.: Lie algebras. New York: Interscience (1962).

\bibitem{Jahn}Jahn, J.: Introduction to the theory of nonlinear optimization. Springer (1996).

\bibitem{Jost}Jost,J., Li-Jost, X.: Calculus of Variations. Cambridge University Press, Vol. 64 (1998).

\bibitem{Kantowski}Kantowski, R., Sachs, R.K.: Some spatially homogeneous anisotropic relativisticcosmological models. J. Math. Phys. 7 (1966)443.

\bibitem{Kempf}Kempf, A., Mangano, G., Mann, R.B.: Hilbert space representation of the minimal length uncertainty relation. Phys. Rev. D 52 (1995) 1108-1118.

\bibitem{Kennard} Kennard, E.H.: Zur Quantenmechanik einfacher Bewegungstypen. Zeitschr. Phys. 44 (1927) 326.

\bibitem{Kober}Kober, M., Nicolini, P.: Minimal scales from an extended Hilbert space. Class. Quant. Grav. 27 (2010) 245024.

\bibitem{Kosinski}Kosi\'nski, P., Bolonek, K.: Minimalisation of uncertainty relations in noncommutative quantum mechanics. Acta Phys. Polon. B34 (2003) 2575-2588.

\bibitem{Leoni}Leoni, G.: A first course in Sobolev spaces. American Mathematical Society (2009).

\bibitem{Madore}Madore, J.: An introduction to noncommutative differential geometry and its physical applications. Cambridge UniversityPress (2002).

\bibitem{Farhoudi}Malekolkalami, B., Farhoudi, M.: Noncommutative double scalar fields in FRW cosmology as cosmical oscillators. Classical Quantum Gravity 27, 245009 (2010).

\bibitem{Matoiu}M$\breve{a}$ntoium M., R. Purice: The magnetic Weyl calculus. J. Math. Phys. 45 (2004) 1394-1417.

\bibitem{Isidro} Monreal, L., Fern\'andez de C\'ordoba, P., Ferrando, A., Isidro, J. M.: Noncommutative space and the low-energy physics
of quasicrystrals. Int. J. Mod. Phys. A 23 (2008) 2037–2045.

\bibitem{Nair}Nair, V.P., Polychronakos, A.P.: Quantum mechanics on the noncommutative plane and sphere. Phys. Lett. B 505 (2001) 267.

\bibitem{Nicolini}Nicolini, P.: Noncommutative black holes, The final appeal to quantum gravity: A review. Int. J. Mod. Phys. A 21 (2009) 1229-1308.

\bibitem{Nijenhuis1}Nijenhuis, A., Richardson, R.W.: Cohomology and deformations in graded Lie algebras. Bull. Amer. Math. Soc. 72 (1966) 1-29.

\bibitem{Nijenhuis2}Nijenhuis, A., Richardson, R.W.: Deformations of Lie algebra structures. J. Math. Mech. 17(1967) 89-105.

\bibitem{Lein} Nittis, G., Lein, M.: Applications of magnetic $\Psi$DO techniques to SAPT. Rev. Math. Phys. 23 (2011) 233-260.


\bibitem{Pfeuffer} Pfeuffer, C., Toft, J.: Compactness properties for modulation spaces, math.FA/1804.00948.

\bibitem{Robertson}Robertson, H.: The uncertainty principle. Phys. Rev. 34 (1929) 163.

\bibitem{Seiberg}Seiberg, N., Witten, E.: String theory and noncommutative geometry. JHEP 9909 (1999) 032.

\bibitem{Shannon}Shannon, C., Weaver, W.: The mathematical theory of communication. University of Illinois Press, Urbana (1949).

\bibitem{Smale}Smale, S.: Differentiable dynamical systems. Bull. Am. Math. Soc. 73 (1967) 747-817.

\bibitem{Szabo}Szabo, R.J.: Quantum field theory on noncommutative spaces. Phys. Rept. 378 (2003) 207-299.

\bibitem{Vilela}Vilela Mendes, R.: Deformations, stable theories and fundamental constants. J. Phys. A: Math. Gen. 27 (1994) 8091-8104.

\bibitem{Weyl}Weyl, H.: Gruppentheorie und Quantenmechanik. Hirzel, Leipzig, (1928).

\end{thebibliography}
\end{document}